%% file: main.tex
\documentclass{article}

\usepackage{microtype}
\usepackage{graphicx}
\usepackage{subfigure}
\usepackage{booktabs} 

\usepackage{hyperref}

\usepackage[accepted]{icml2024}

\usepackage{amsmath}
\usepackage{amssymb}
\usepackage{mathtools}
\usepackage{amsthm}

\usepackage[capitalize,noabbrev]{cleveref}

\theoremstyle{plain}

\usepackage{parskip}
\usepackage[utf8]{inputenc} 
\usepackage[T1]{fontenc}    
\usepackage{hyperref}       
\hypersetup{colorlinks,allcolors=blue,linktocpage=true}
\usepackage{url}            
\usepackage{booktabs}       
\usepackage{amsfonts}       
\usepackage{nicefrac}       
\usepackage{microtype}      
\usepackage{xcolor}         
\usepackage{mathtools}
\usepackage{amsthm, bbm, bm}
\usepackage{thm-restate, thmtools, enumitem}
\usepackage{soul}
\usepackage{multicol}

\newtheorem{lemma}{Lemma}

\newtheorem{theorem}{Theorem}
\newtheorem{assumption}{Assumption}
\newtheorem{definition}{Definition}

\newtheorem{property}{Property}
\newtheorem{remark}{Remark}
\newtheorem{prop}{Proposition}

\input{macros}

\usepackage{enumitem}

\usepackage[textsize=tiny]{todonotes}

\icmltitlerunning{Kernel Debiased Plug-in Estimation}

\begin{document}

\twocolumn[
\icmltitle{Kernel Debiased Plug-in Estimation: Simultaneous, Automated Debiasing without Influence Functions for Many Target Parameters}




\begin{icmlauthorlist}
\icmlauthor{Brian Cho}{yyy}
\icmlauthor{Yaroslav Mukhin}{comp}
\icmlauthor{Kyra Gan}{yyy}
\icmlauthor{Ivana Malenica}{sch}
\end{icmlauthorlist}

\icmlaffiliation{yyy}{Department of ORIE, Cornell Tech, NY, USA}
\icmlaffiliation{comp}{Massachusetts Institute of Technology, Cambridge, MA, USA}
\icmlaffiliation{sch}{Department of Statistics, Harvard University, Cambridge, MA, USA}

\icmlcorrespondingauthor{Brian Cho}{bmc233@cornell.edu}
\icmlkeywords{Machine Learning, ICML}

\vskip 0.3in
]



\printAffiliationsAndNotice{}  

\begin{abstract}
When estimating target parameters in nonparametric models with nuisance parameters, substituting the unknown nuisances with nonparametric estimators can introduce ``plug-in bias.'' Traditional methods addressing this suboptimal bias-variance trade-off rely on the \emph{influence function} (IF) of the target parameter. When estimating multiple target parameters, these methods require debiasing the nuisance parameter multiple times using the corresponding IFs, which poses analytical and computational challenges. In this work, we leverage the \emph{targeted maximum likelihood estimation} (TMLE) framework to propose a novel method named \emph{kernel debiased plug-in estimation} (KDPE). KDPE refines an initial estimate through regularized likelihood maximization steps, employing a nonparametric model based on \emph{reproducing kernel Hilbert spaces}. We show that KDPE:
(i) simultaneously debiases \emph{all} pathwise differentiable target parameters that satisfy our regularity conditions,
(ii) does not require the IF for implementation, and
(iii) remains computationally tractable. 
We numerically illustrate the use of KDPE and validate our theoretical results.

\end{abstract}

\section{Introduction}
\input{introduction}

\subsection{Related Works} \label{sec:related_work}
\input{related_work}
\section{Problem Setup and Preliminaries} \label{sec:problem_setup}
\input{problem_setup}

\subsection{RKHS Preliminaries}
\label{subsec:rkhs_prelim}
\input{rkhs_preliminary}

\section{Debiasing with RKHS} \label{sec:rkhs}
\input{rkhs}

\section{Implementation/Simulation Results} 
\label{sec:sims}
\input{simulations}

\section{Conclusions and Practical Benefits}
\input{discussion}

\bibliography{citation}
\appendix

\thispagestyle{empty}

\onecolumn 
\section*{Appendix}
\section{Notation} \label{app:notation}
\input{notation}
\newpage

\section{KDPE Theory} \label{app:theorem_main}
\input{appendixB}
\section{Simulation Details and Additional Empirical Results } \label{app:simulation}
\input{appendixD}

\end{document}

%% file: macros.tex
\DeclarePairedDelimiter\norm{\lVert}{\rVert}
\DeclarePairedDelimiterX\brk[2]{\langle}{\rangle}{#1\,,\,#2} 
\DeclarePairedDelimiter\set{\{}{\}} 
\DeclarePairedDelimiterX\Set[2]{\{}{\}}{#1 \;;\; #2} 

\newcommand{\s}[1]{^{(#1)}}
\newcommand{\h}{\widehat}

\newcommand{\argmax}{\operatornamewithlimits{arg\,max}}
\newcommand{\argmin}{\operatornamewithlimits{arg\,min}}
\newcommand\independent{\protect\mathpalette{\protect\independenT}{\perp}}
\def\independenT#1#2{\mathrel{\rlap{$#1#2$}\mkern2mu{#1#2}}}
\newcommand\indep\independent

\DeclareMathOperator{\vspan}{span}

\newcommand{\NN}{\mathbb{N}}

\newcommand{\PP}{\mathbb{P}}

\newcommand{\RR}{\mathbb{R}}

\def \mR     {\mathbb{R}}

\def \mN     {\mathbb{N}}
\def \mP     {\mathbb{P}}

\def \mE     {\mathbb{E}}

\newcommand{\Acal}{\mathcal{A}}

\newcommand{\Hcal}{\mathcal{H}}

\newcommand{\Mcal}{\mathcal{M}}

\newcommand{\Ocal}{\mathcal{O}}

\newcommand{\Xcal}{\mathcal{X}}
\newcommand{\Ycal}{\mathcal{Y}}

\newcommand{\calB}{\mathcal{B}}
\newcommand{\calC}{\mathcal{C}}

\newcommand{\calF}{\mathcal{F}}

\newcommand{\calH}{\mathcal{H}}

\newcommand{\calM}{\mathcal{M}}

\newcommand{\calO}{\mathcal{O}}

\newcommand{\calX}{\mathcal{X}}

\def \a      {\alpha}

\def \e      {\epsilon}

\def \s      {\sigma}

\def \w      {\omega}

\def \l      {\lambda}

\newcommand{\Bal}{\begin{align}}
\newcommand{\Eal}{\end{align}}
\newcommand{\Beq}{\begin{equation}}
\newcommand{\Eeq}{\end{equation}}
\newcommand{\Bit}{\begin{itemize}}
\newcommand{\Eit}{\end{itemize}}
\newcommand{\Ben}{\begin{enumerate}}
\newcommand{\Een}{\end{enumerate}}

\newcommand{\Ba}{\begin{array}}
\newcommand{\Ea}{\end{array}}

\newcommand{\Bvec}{\left(\begin{array}{c}}
\newcommand{\Evec}{\end{array}\right)}

\newcommand{\Bmat}{\left(\begin{array}}
\newcommand{\Emat}{\end{array}\right)}

\newcommand{\Bol}{\begin{outline}}
\newcommand{\Eol}{\end{outline}}

\newcommand{\ATE}{\mathrm{ATE}}
\newcommand{\ORr}{\mathrm{OR}}
\newcommand{\RRm}{\mathrm{RR}}

\newcommand{\KDPE}{\text{KDPE}}
\newcommand{\TMLE}{\text{TMLE}}

\def\mbf#1{\mathbf{#1}}	    

\def\tsf#1{\textsf{#1}}

%% file: introduction.tex



Estimating a target parameter $\psi(P)$ from $n$ independent samples of an unknown distribution $P$ is a fundamental and
rapidly evolving statistical field.
Driven by data-adaptive machine learning methods and efficiency theory, modern approaches to estimation can achieve optimal performance under minimal assumptions on the true data-generating process (DGP). These statistical learning advancements enable researchers to obtain strong theoretical guarantees and empirical performance while avoiding unrealistic assumptions.

Modern 
estimation
methods 
rely on
the plug-in principle \citep{van2000asymptotic}, which 
substitutes
unknown parameters of the underlying data-generating process with estimated empirical counterparts
(e.g., the sample mean is 
the average of observations over the empirical distribution $\PP_n$ defined by the dataset). Flexible, machine learning (ML) estimation methods have further exploited the plug-in approach. Meta-learners in the causal inference literature use random forests and neural networks to estimate empirical conditional density/regression functions for the downstream task of estimating target parameters such as potential outcome means and treatment effects \citep{K_nzel_2019}. The use of highly adaptive, complex ML algorithms, however, induces \emph{plug-in bias} (first-order bias) that impacts the downstream estimate.

Estimators such as double machine learning (DML) \cite{chernozhukov2017double}, targeted maximum likelihood estimation  
\cite{van2006targeted}, and the one-step correction \cite{bickel998} remove 
plug-in bias and achieve efficiency (i.e., lowest possible asymptotic variance) within the class of \emph{regular and asymptotically linear} (RAL) estimators.
However,
these methods typically require knowledge of the influence function or the efficient influence function (EIF)\footnote{{The notation of EIF arises when we are in the semi-parametric setting where the directions in which we can perturb the data generating distribution are restricted. Otherwise, IF is the EIF.}} of the target parameter. Deriving and computing these functions analytically 
is often
challenging for practitioners \citep{carone2018toward, jordan2022empirical, kennedy2023} due to their dependence 
on both the statistical functional of interest and the statistical model (e.g., assumptions on the DGP). For example, even seemingly innocuous target parameters, such as the average density value, require specialized knowledge to study with conventional techniques \cite{carone2018toward}. 
Further,
the knowledge of the IF/EIF is specific to a single target parameter under a particular
statistical model. 
When estimating multiple target parameters, existing methods require the derivation of the IF/EIF
for \emph{each target parameter},
hindering the adoption of these techniques in practice \cite{hines2022}.


\paragraph{Contributions} 
We propose a generic approach for constructing efficient RAL estimators that combines the TMLE \cite{one_step_tmle} framework with a novel application of the reproducing kernel Hilbert spaces (RKHSs). 
Leveraging the universal approximation property of 
RKHSs,
we construct a debiased distribution (or nuisance) estimate $P^\infty_n$ that has \emph{negligible plug-in bias} and \emph{attains efficiency} under similar assumptions to existing approaches. Contrary to popular approaches, our kernel debiased plug-in approach
\begin{enumerate}[leftmargin=*, itemsep=0pt, parsep=0pt]
    \item offers an \emph{automated framework} that does not require the computation of (efficient) influence function $\phi_P$, debiased/orthogonalized estimating equations, or efficiency bounds from the user, and
    \item the same nuisance/distribution estimate $P^\infty_n$ can be used as a plug-in estimate for \emph{all} pathwise differentiable target parameters $\psi(P)$ that satisfy some standard regularity conditions (provided in Section~\ref{sec:rkhs}).
\end{enumerate}
The technical contribution of our work is to identify the set of sufficient conditions to eliminate the plug-in bias under our proposed framework. In particular,  Assumption~\ref{assump:S5} is our novel regularity assumption on the model $\mathcal{M}$ and functional $\psi$ needed to control the plug-in bias term 
of the proposed KDPE estimator. We numerically validate our results and illustrate that 
our method performs as well as state-of-the-art modern causal inference methods, which explicitly use the functional form of the IF.


{The paper is organized as follows. In Section~\ref{sec:problem_setup}, we describe the problem setup and introduce key concepts related to TMLE and RKHSs.
To ease presentation, 
we focus on
the most generic case where 
1) the estimation problem is nonparametric (i.e., no assumption on the DGP), 
2) the nuisance is considered to be the entire distribution $P$, and
3) the target parameter $\psi$ is a mapping from the probability distribution $P$ to real numbers.
We extend 
our procedure
to 
semiparametric models and other nuisances
in Appendix~\ref{app:DGP12}. 
In Section~\ref{sec:rkhs}, we propose \emph{kernel debiased plug-in estimator},
and
characterize the assumptions needed for
KDPE to be asymptotically linear and efficient.
In Section~\ref{sec:sims}, we provide concrete empirical examples of KDPE for target parameters such as the average treatment effect, odds ratio, and relative risk, in both single-time period and longitudinal settings.}


%% file: related_work.tex
Our work is related to two types of approaches that can debias a target parameter:
1) IF-based, 
and 2) IF-free computational approaches. 
When compared with these two approaches, \emph{KDPE simultaneously debiases multiple target parameters} and does not require the derivation of IF. 
In addition, KDPE relates to 
undersmoothed \emph{highly adaptive lasso}- (HAL-) MLE, which can be used to simultaneously debias multiple parameters. 
When compared with HAL-MLE, our method is computationally more efficient.

The main techniques for constructing IF-based efficient RAL estimators include estimating equations, 
DML,
one-step correction,
and TMLE.
A general estimating equation approach involves solving
for the target parameter by setting the score equations to zero \citep{robins1986, newey1990, bickel998}.  
Double machine learning 
characterizes the estimand as the solution to a population score equation \citep{chernozhukov2017double}. 
One-step correction method adds
the empirical average of the IF to an initial estimator \citep{bickel998}. We 
defer the discussion on TMLE to
Section \ref{subsec::tmle}.

Current IF-free methods for efficient estimation include a) approximating IF through finite-differencing~\citep{carone2018toward, jordan2022empirical} or Monte Carlo \citep{agrawal2024automated},  and
b) AutoDML~\cite{chernozhukov2022automatic}.
In contrast to methods in a), KDPE 
does not attempt to approximate the IF of a target parameter in any way.
%
AutoDML, on the other hand, relies on
%
an orthogonal reparameterization of the problem to achieve debiasing.
It automates the estimation of the additional nuisance parameter by exploiting the structure of the estimating equations in the particular setting of \citet{chernozhukov2022automatic}.  
KDPE, in contrast, directly considers the plug-in bias and eliminates this term from the final estimate.

The HAL-MLE \citep{laan_undersmooth2022} 
solves the score equation for \emph{all} cadlag functions with bounded L1 norm (which is assumed to include the desired EIF). Like HAL, KDPE solves a rich set of scores, rather than a single targeted score,  to approximately solve the score equation at a $\sqrt{n}$-rate. 
HAL, however, is highly computationally inefficient for large models due to its basis functions growing exponentially with the number of covariates $d$ and polynomially with the sample size $n$. 
KDPE uses only $n$ basis functions, improving computational tractability to solve all scores in a \emph{universal} RKHS.

%% file: problem_setup.tex
%
Let $O_1, \ldots, O_n$ denote iid observations drawn from a probability
distribution $P^*$ on the sample space $\calO$.
The distribution $P^*$ is \emph{unknown} but is assumed to belong to a
\emph{nonparametric}\footnote{We provide a simple semiparametric extension, where we restrict the data perturbation directions, 
in Appendix~\ref{app:DGP12}.} collection $\Mcal$, which consists of distributions on $\calO$ dominated by a common measure.
Let $\psi:\calM\rightarrow\RR^m$ be a functional of the model $\Mcal$, also
referred to as the \emph{target parameter}.
Our goal is to ``efficiently'' estimate $\psi(P^*)$ on the basis of the $n$ observations using a plug-in estimator (i.e., by estimating $P^*$
and evaluating $\psi$ on this estimate) \emph{without using influence functions.}
Since our method produces a single debiased estimate $P^\infty_n$ that is \emph{independent} of target parameter $\psi$, we assume that $m=1$ wlog.

\paragraph{Notation} 
The density of $P$ is denoted by the lowercase $p$, which we use interchangeably to refer to the probability measure.
We use $Pf$ and $P[f]$
to indicate the expectation $\mE_P[f] = \int f \, dP$ and $\PP_n$ to denote the
empirical measure $\PP_nf := \frac{1}{n}\sum_{i=1}^n f(O_i)$.
Let $L^2_0(P)$ be the set of all 
square-integrable functions with respect to probability measure $P$, i.e.,
$f:\calO\to\mR$ with
$Pf=0$ and $Pf^2<\infty$. A complete notation table is provided in Appendix~\ref{app:notation}.

To formalize the notation of asymptotic efficiency,  we first introduce RAL estimators in Section~\ref{subsec:RAL_estimator}. In Section~\ref{subsec:plug_in_bias}, we introduce \emph{plug-in bias} and
provide 
the 
assumptions and condition required for plug-in estimators to be \emph{asymptotically linear} and \emph{efficient}. 
We introduce TMLE in Section~\ref{subsec::tmle} 
since KDPE shares the same construction framework as TMLE, with one main modification. Finally, we introduce the necessary concepts related to RKHSs in Section~\ref{subsec:rkhs_prelim}, 
as KDPE employs a RKHS-based model.
\subsection{Regular and Asymptotic Linear Estimators}\label{subsec:RAL_estimator}
%
\textit{Asymptotic linearity} of an estimator
leads to a tractable limiting distribution, resulting in asymptotically valid inference.
It corresponds to the ability to approximate the difference between an estimate and the true value of the target parameter as an average of i.i.d.
random variables. 
\begin{definition}[Asymptotic linearity]\label{defn:LAN}
    An estimator $\hat{\psi}_n$ is asymptotically linear if $\hat{\psi}_n - \psi(P^*) = \PP_n\phi_{P^*} + o_{P^*}(1/\sqrt{n})$,
     where  $\phi_{P^*}:\Ocal \rightarrow \RR^m$ is the corresponding \textit{influence function} for estimator $\hat{\psi}_n$.
\end{definition}
On the other hand, a \textit{regular} estimator attains the same limiting distribution even under perturbations (on the order of $\sqrt{n}$) of the true data-generating distribution, enforcing robustness to distributional shifts.
Since the most efficient (i.e., smallest sampling variance) regular estimator is guaranteed to be asymptotically linear \citep{van1991differentiable},  we restrict our attention to the class of RAL estimators.\footnote{Our nonparametric assumption on the model class $\Mcal$ implies that all RAL estimators have the IF.}
\subsection{Plug-in 
Bias}\label{subsec:plug_in_bias}
To connect plug-in estimators (i.e., $\hat{\psi}_n \equiv \psi(\h{P})$, where $\h{P}$ is our distributional estimate with $n$ samples) and RAL estimators, we consider target parameters that satisfy
\emph{pathwise differentiability} \citep{koshevnik1976non}. This means that given a smoothly parametrized one-dimensional submodel
$\{P_\e\}_{0^- <\e < 0^+}$ of $\calM$, 
the map $\e \mapsto \psi(P_\e)$ is differentiable in the ordinary sense,\footnote{For a concrete definition of the one-dimensional sub-model $\{P_\epsilon\}_{0^- < \epsilon < 0^+}$, we refer readers to Equation \eqref{eq:score_model}, which gives an explicit submodel $p_{\epsilon, h}$. There exist many choices for this sub-model; we opt for the linear model $p_{\epsilon, h}$ here for simplicity.}
and the derivative has a Riesz representation discussed below. Pathwise differentiability of a target parameter admits \emph{regular} estimators that 
(i) converge at the parametric $\sqrt{n}$-rate and 
(ii) whose asymptotic normal distributions vary smoothly with $P$.
To see this, we first introduce the tangent space. 

\paragraph{Scores} 
Let $h = {\frac{d}{d\e}}_{|\e=0} \log p_\e$ be the \emph{score} of the path at
$P_0$ (i.e., $P_{\e=0}$).
Then, the smoothness requirement on the path implies that the expected value of the score function under distribution $P_0$ equals $0$, i.e.,
$P_0 h = 0$, and $P_0 h^2 < \infty$ \cite{van2000asymptotic}.
The collection of scores of all smooth paths through $P_0$ forms a linear space, which we refer to as the \emph{tangent space} of $\calM$ at $P_0$.
Assuming that $\calM$ is nonparametric, for any $P \in \Mcal$, its tangent space equals $L^2_0(P)$.\footnote{Intuitively, the tangent space at distribution $P$ contains all directions through which we can move from the current distribution $P$ to another distribution in model class $\mathcal{M}$.}
Given any distribution $P\in\calM$ and a score $h$, we work with the concrete path\footnote{Intuitively, the set of $\{P_{\epsilon, h}\}_{\epsilon, h}$ can be thought as a reparameterization of the model class $\mathcal{M}$ centered around a local distribution $P$. We refer to Appendix \ref{app:questions_from_GgGE} for further questions regarding the validity of our linear submodel.}
\begin{align} \label{eq:score_model}
    p_{\e,h} 
    \coloneqq
    p(\e | h) \coloneqq  (1+\e h)p.
\end{align}

\paragraph{Influence function} 
While the IF
defines the limiting sampling variance of a RAL estimator in Definition \ref{defn:LAN}, it has additional, related interpretations: 
the IF is the Riesz representer of the derivative functional. Pathwise differentiability of $\psi$ at $P$  implies that there exists a continuous linear functional\footnote{Linear with respect to $h$.} $D_{P}\psi:L^2_0(P)\rightarrow \mR$
such that for any smooth path with a score of $h$, the derivative of the functional $\psi(\cdot)$ at distribution $P$, ${\frac{d}{d\epsilon}}_{|\e=0}\psi(P_{\epsilon,h})$, is equivalent to evaluating  $D_{P}\psi[\cdot]$ at score $h$, i.e., ${\frac{d}{d\e}}_{|\e=0}\psi(P_{\e,h}) = D_{P}\psi[h]$.
The \emph{influence function} $\phi_P \in L^2_0(P)$ of parameter $\psi$ is the Riesz representer of its derivative $D_P\psi$ for the $L^2_0(P)$ inner product and is characterized by the property that $D_P\psi[h] = \brk{h}{\phi_P}_{L^2_0(P)}$ for every $h\in L^2_0(P)$. 

Next,  Lemma~\ref{lemma:expension}\citep{bickel998} decomposes the estimation error,  $\psi(\h{P})-\psi(P^*)$, by considering the following concrete path between $\hat P$ and $P^*$: $\hat p(\e | h) = (1-\e)\hat p + \e p^*$. We note that under this path, the direction $h$ becomes $h = p^*/p - 1$.
By leveraging the Riesz Representation property, we obtain the error decomposition shown below:
 %
%
%
\begin{lemma}[\citealt{bickel998}] \label{lemma:expension} 
  Let $\psi:\calM\to\mR$ be pathwise differentiable and $\h P, P^*\in\calM$.
  Assume that $\h P$ satisfies:
  (1) $P^*$ is absolutely continuous with respect to $\h P$, and
  (2) Radon–Nikodym derivative $dP^*/d\h P$ is square integrable under $\h P$.
  Then the following von Mises expansion holds:
  \begin{align}  \label{eq:asyExpansion} 
      \psi(\h{P})-\psi(P^*) 
      &=
      \PP_n\phi_{P^*} - \PP_n\phi_{\h{P}} 
      \\&+
      (\PP_n-P^*)\big[\phi_{\h{P}}-\phi_{P^*}\big] 
      +
      R_2(\h{P},P^*)\notag
    \end{align} 
\noindent where $R_2(\h{P},P^*)$ is
  the
  second-order (quadratic) remainder in the difference between $\h{P}$ and
  $P^*$. 
\end{lemma}

The expansion in Lemma \ref{lemma:expension} closely resembles Definition~\ref{defn:LAN}. It
motivates the following common
assumptions that control the behavior of the empirical process term $(\PP_n-P^*)\big[\phi_{\h{P}}-\phi_{P^*}\big]$ and the second-order reminder $R_2(\h{P},P^*)$
\citep{van2006targeted}.

\begin{assumption}\label{assump:empprocess} 
  Assume that (i) $\norm{\phi_{\h P} - \phi_{P^*}}_{L^2(P^*)} = o_{P^*}(1)$
  and that (ii) there exists an event $\Omega$ with $P^*(\Omega)=1$ such that the set
  of functions $\Set{\phi_{\h P_n(w)}}{w\in\Omega, n\in \mN}$ is
  $P^*$-Donsker. 
  Then the empirical process term satisfies:
  $ 
  (\mP_n-P^*)\big[\phi_{\h P}-\phi_{P^*}\big]
      =
      o_{P^*}(1/\sqrt{n}).
  $ 
\end{assumption} 
\begin{assumption}\label{assump:secondorderremainder} 
  Assume that $\h P$ converges sufficiently fast and $\psi$ is sufficiently
  regular so that the second-order remainder term
  $R_2( \h{P}, P^*)$ is $o_{P^*}(1/\sqrt{n})$.
\end{assumption} 
We note that Assumptions~\ref{assump:empprocess} and \ref{assump:secondorderremainder} are standard in the TMLE literature~\cite{van2006targeted}.
The goal of this work is to construct a plug-in estimator
$\psi(\widehat P)$ that is asymptotically linear and converges at the parametric (efficient) $\sqrt{n}$-rate to the Gaussian asymptotic distribution $N(0, P^*[\phi_{P^*}^2])$.
Since $\PP_n\phi_{P^*}$ governs the asymptotic distribution of $\psi(\widehat P)$,
Assumptions~\ref{assump:empprocess} and \ref{assump:secondorderremainder} leave us with a single term that must converge at an $o_{P^*}(1/\sqrt{n})$ rate:
$$\psi(\h{P})-\psi(P^*) = \PP_n\phi_{P^*} - \PP_n\phi_{\h{P}} + o_{P^*}(1/\sqrt{n}). 
$$
We denote this remaining term $\PP_n \phi_{\widehat P}$ as the \emph{plug-in bias}. 
In contrast, a naive plug-in estimator may have a first-order/plug-in bias term that dominates the $\sqrt{n}$-asymptotics.
\subsection{TMLE Preliminaries}\label{subsec::tmle} 
We briefly recap TMLE~\citep{van2006targeted}, as its construction closely resembles that of our estimator. TMLE is a plug-in estimator that satisfies Definition~\ref{defn:LAN} and
is constructed in two steps.
First, one obtains a preliminary estimate $P^0_n$ of $P^*$, which is typically consistent but not efficient (slower than $\sqrt{n}$-rate). 
The TMLE solution to debias $P^0_n$ for the parameter $\psi$ is to perturb
$P^0_n$ in a way that 
(i) increases the sample likelihood of the distribution estimate, and 
(ii) sets the plug-in bias $\mP_n\phi_{\h P} = 0$, maintaining consistency and mitigating first-order bias. \citet{van2006targeted} consider the parametric model in
\eqref{eq:score_model} with the score equal to the influence function
$\phi_{P^0_n}$ and update the estimate to be the MLE
$p^1_n = (1+\epsilon_n \phi_{P^0_n})p^0_n$,
where
\begin{align} \label{eq:TMLE} 
    \e_n
    &\coloneqq
    \argmax_{\e: p^0_n(\e) \in \calM}
    \sum_{i=1}^n \log \left(1 + \e \phi_{P^0_n}(O_i)\right)
        p^0_n(O_i).
      \end{align} 
Assuming that we converge to an interior point such that first-order conditions hold, 
the updated estimate $p^1_n$ solves the following: 
 \begin{align} \label{eq:EffScoreEstimEq} 
    0 = \PP_n \frac{\frac{d}{d\e} p^0_n(\e)}{ p^0_n(\e)}
    =
    \frac{1}{n}\sum_{i=1}^n \frac{
        \phi_{P^0_n}(O_i)
    }{
        1+\e \phi_{P^0_n}(O_i)}
        = \PP_n[\frac{
        \phi_{P^0_n}
    }{
        1+\e \phi_{P^0_n}
    }].
  \end{align} 
By iterating the TMLE step, we get a sequence of updated
estimates $\{p^\ell_n\}_{\ell=0}^\infty$.
Assuming that $p^\ell_n$ converges as $\ell\to\infty$ sufficiently regularly,\footnote{We have omitted details from the main exposition of TMLE to provide a brief overview of the method. We refer to the last paragraph on Page 11, \cite{van2006targeted}, for the full characterization.} the limit
$p^\infty_n$ is a \emph{fixed point} of \eqref{eq:TMLE}, the corresponding
argmax in \eqref{eq:TMLE} is
$\e_n =0$, 
and
\eqref{eq:EffScoreEstimEq} becomes the plug-in bias term: $\PP_n \phi_{P^\infty_n}=0$.
This property of $P^\infty_n$ achieves asymptotic linearity and efficiency for the plug-in estimator $\psi( P^\infty_n)$ under Assumptions \ref{assump:empprocess} and \ref{assump:secondorderremainder}.

We construct KDPE following
the TMLE framework with one 
modification. Instead of reducing $\calM$ into a one-dimensional submodel by perturbing $P^0_n$ in the direction of $\phi_{P_n^0}$, we construct RKHS-based submodels.
These submodels are of infinite dimensions and are independent of both $\phi_P$ and $\psi$.  Next, we introduce RKHS properties.

%
%
%
%

%% file: rkhs_preliminary.tex
Let $\Xcal$ be a non-empty set.
A function $K: \Xcal \times \Xcal \rightarrow \RR$ is called
positive-definite (PD) if for \emph{any} finite sequence of inputs 
$\mbf{x}= [x_1, ..., x_n]^\top \in \Xcal^n$, the matrix 
$ 
K_{\mathbf{x}}
\coloneqq
[K(x_i, x_j)]_{ij}
$ 
is symmetric and positive-definite~\citep{micchelli2006universal}. 
The \emph{kernel function at $x$} is the univariate map 
$y\mapsto k_x(y) \coloneqq K(x,y)$;
it is common to overload the notation with vector-valued maps
$k_\mathbf{x} \coloneqq [k_{x_1}, \ldots, k_{x_n}]^\top$; letting $\mathbf{y}= [y_1, ..., y_m]^\top$ be a column vector, we define
the \emph{kernel matrix}
$K_{\mathbf{xy}} \coloneqq [K(x_i, y_j)]_{ij} \in \mR^{n\times m}$ and
$K_\mbf{x}\coloneqq K_\mbf{xx}$. 

The \emph{reproducing kernel Hilbert space} associated with kernel $K$, denoted $\Hcal_K$, is a unique Hilbert space of functions on $\Xcal$ that satisfies the following properties: (i) for any $x\in \Xcal$, the function $k_x \in \Hcal_K$, and (ii) $\langle f, k_x \rangle_{\Hcal} = f(x)$ for all $f \in \Hcal_K$. 
RKHSs have two key features that facilitate our results:
(i) suitable choices of RKHSs, designated as universal kernels, are sufficiently rich to approximate any influence function, and 
(ii) optimization of sample-based objectives (e.g., likelihood) over the
RKHS-based models can be efficiently carried out using a representer-type theorem,
allowing for a tractable solution.

Denoting $\calC_0(\calX)$ as the space of continuous functions that vanish at
infinity with the uniform norm,
we say that a kernel $K$ is $\calC_0$ provided that $\calH_K$ is a subspace of
$C_0(\calX)$.
We say that a $\calC_0$-kernel $K$ is  \emph{universal} \citep{carmeli2010vector} if it satisfies the following property:


\begin{property}[Universal kernels] \label{def:uk_cf} 
  Let $\Xcal$ be a closed subset of $\mR^n$, and let $K$
  be a PD kernel such that the map $x\mapsto K(x,x)$ is bounded and 
  $k_x \in \calC_0(\calX)$ for all $x\in\calX$.
  Then $\calH_K \subset \calC_0(\calX) \subset L^p(P) \subset L^q(P)$
  for all $1\le p<q\le\infty$ and every probability measure $P$. Furthermore, the kernel $K$ is universal if and only if the space $\calH_K$ is dense in $L^p(\calX, P)$ for at least one (and consequently for all) $1\le p<\infty$ and each probability measure $P$ on $\calX$.
\end{property} 

Property~\ref{def:uk_cf} implies that universal kernels contain sequences that can approximate the influence function $\phi_P$ arbitrarily well with
respect to the $L^2(P)$ norm for any $P$. Universal kernels include the  \emph{Gaussian kernel} $K(x,y) =  \exp(-\|x-y\|_2^2)$, which is the primary example in our work. Our debiasing method relies on constructing local (to $P$) submodels of $\calM$ indexed by a set of scores at $P$.
Since scores are $P$-mean-zero, 
given any PD kernel $K$ on $\calX$ that is bounded, $C_0$, and universal, we transform $K$ into a mean-zero kernel with respect to $P$
(proof in Appendix~\ref{proof:mean-zero-kernel}):
\begin{prop}[Mean-zero RKHS] \label{thm:mean-zero-kernel} 
  Let
  $
    \calH_P \coloneqq \Set{h\in\calH_K}{Ph = 0}
  $
  be the subspace of $\calH_K$ containing all $P$-mean-zero functions.
  Then 1) $\calH_P= \calH_{K_P}$ is closed in $\calH_K$ and is also an RKHS, with reproducing kernel
  $
  K_P(x,y) 
  \coloneqq
  K(x,y) 
  - 
  \frac{
    \int_{\Xcal}K(x,s)dP(s)\int_{\Xcal}K(y,s)dP(s)
  }{
    \int_{\Xcal \times \Xcal}K(s,t)dP(s)dP(t)
  }
  $, and 2) $\calH_P$ is dense in $L^2_0(P)$.
\end{prop} 

%% file: rkhs.tex
We propose a novel distributional estimator $P_n^\infty$, 
which has the following properties:
(i) it takes two user inputs: a pre-estimate $P^0_n$ and an RKHS $\calH_K$
(equivalently, the kernel $K$) on the sample space $\calO$;
(ii) it provides a debiased and efficient plug-in $\psi(P^\infty_n)$ for
estimating any parameter $\psi$ under the stated regularity conditions;
(iii) it solves the influence curve estimating equation asymptotically, thereby eliminating the plug-in bias of $P^0_n$ but
unlike TMLE this holds simultaneously for all
pathwise differentiable parameters; 
(iv) it does not require the influence function $\phi_{P}$ to be
implemented and does not depend on any $\psi$. 
We now describe each step of this estimator in detail. 

\textbf{Local fluctuations of pre-estimate$\;$} 
Let $\calH_K$ be an RKHS with a bounded kernel $K$ on the sample space $\calO$
and $\calH_P$ be the subspace of scores at $P$ within
$\calH$.
For each $P\in\calM$,
\begin{align} \label{eq:hcalmodel}
  \widetilde\calM_{P}
  \coloneqq &
  \{
    p_{h} \coloneqq p(h) \coloneqq (1+h)p  ; \\ &
    h\in\calH_{P} \text{ such that } (1+h)p > 0
  \} \cap \calM.\notag
\end{align}
This submodel\footnote{We provide further details regarding the RKHS-constrained model $\widetilde{\mathcal{M}}_P$ and its interpretation in Appendix \ref{app:questions_from_GgGE}.} allows $P$ to be perturbed along any direction $h\in\calH_P$ that
results in a valid distribution in the model $\calM$.
These directions $h$ must have zero mean under $P$, 
requiring
modifying
the kernel $K$ as described in Proposition~\ref{thm:mean-zero-kernel}.

\textbf{Regularized MLE$\;$} 
To choose the debiasing update from $\widetilde\calM_P$, we follow TMLE (Section~\ref{subsec::tmle}) and maximize the empirical likelihood.
However, since our fluctuations are indexed by the infinite-dimensional set
$\calH_P$, the MLE problem is generally not well-posed \citep{fukumizu2009} and
must be regularized.
This is achieved with an RKHS-norm
penalty,
leading to the following optimization problem as a mechanism to choose the
update for an estimate $P$:
%
\begin{align} \label{prob:kMLE} 
  \argmin_{h\in\calH_P : p(h) \in \widetilde\calM_{P}}
  \mP_n \big[ -\log (1+h)p \big]
  +
  \l \norm{h}^2_{\calH_P},
\end{align} 
where the regularization parameter $\l\ge 0$  regularizes the complexity of
the perturbation $h$ to $P$.
Note that \eqref{prob:kMLE} is an infinite-dimensional optimization problem.
Next, Theorem~\ref{thm:representer} 
states that the solution to \eqref{prob:kMLE} is guaranteed to be
contained in an explicitly characterized finite-dimensional subspace, simplifying the optimization problem.
Let $\calH_P$ be an RKHS of scores at $P$ with a 
mean-zero
PD
kernel
$K_P$ on the sample space $\calO$,
and $\mbf{O}=[O_1,\ldots, O_n]^\top\in\calO^n$ be distinct data points. Define the $n$-dimensional subspace
$ 
  S_n
  \coloneqq
  \vspan\set{k_{O_1}, \ldots, k_{O_n}}\subset\calH
$ 
and let
$\pi_n:\calH\to S_n$ be the orthogonal projection onto $S_n$.
\begin{theorem}[RKHS representer] \label{thm:representer} 
Assume that
$ \widetilde\calM_P$
is closed under the
projection onto $S_n$:
whenever 
$p(h)\in\widetilde\calM_P$,
so is 
$p(\pi_n h)\in\widetilde\calM_P$.
Then for any loss function of the form 
$
  L_{p,h}(O) = \tilde L(O,p(O),h(O)),
$
the minimizer of the regularized empirical risk 
\begin{align} \label{prob:ermH} 
  h^* \in 
  \argmin_{h\in\calH : p(h)\in\widetilde\calM_P}
  \mP_n L_{p,h} + \l\norm{h}^2_{\calH}
\end{align} 
admits the following representation:
$ 
  h^*
  =
  \mbf{\a}^\top k_\mbf{O} + \check h,
  $ 
for some $\mbf\a\in\mR^n$,
where, if $\l=0$, then $\check h\in\calH$ is any such that
$\check h(O_i)=0$ for each $i\in [n]$,
and, if $\l>0$, then $\check h\equiv 0\in\calH$.
\end{theorem} 


Our choice of loss function in \eqref{prob:kMLE}, $L_{p,h}(O) = -\log\left((1+h(O))p(O)\right)$, satisfies the loss function conditions of Theorem \ref{thm:representer}, and therefore admits a solution of the form $h = \alpha^\top k_{\mathbf{O}} + \check{h}$. Thus, Theorem~\ref{thm:representer} implies that in order to minimize \eqref{prob:kMLE}
over the perturbations $\widetilde\calM_P$,
it suffices to consider the minimization problem over the $n$-dimensional submodel
\begin{align} \label{eq:Mn} 
  \widetilde\calM_P \supset
  \widetilde\calM_{P,n}
  \coloneqq & 
  \{
    p(\mbf{\a}^\top k_\mbf{O}) \coloneqq 
    (1 + \mbf{\a}^\top k_\mbf{O})p \ ; \\
  &
    \mbf{\a}\in\mR^n \text{ such that } p(\mbf{\a}^\top k_\mbf{O}) > 0
  \} \cap \calM,
  \notag 
\end{align} 
leading to the following simpler (parametric) MLE:
\begin{align} \label{eq:kMLEsol} 
  &\mbf{\a}_n(p) \coloneqq
  \mbf{\a}_n(p;\l)
  \coloneqq \\
  &\argmin_{\mbf{\a}\in\mR^n: p(\mbf{\a}^\top k_\mbf{O})\in\widetilde\calM_{P,n}} \mP_n \big[
    -\log (1 + \mbf{\a}^\top k_\mbf{O})p
  \big]
  +
  \l\mbf{\a}^\top K_\mbf{O} \mbf{\a}.\notag
\end{align} 


\begin{remark}[Regularization]\label{remark:regularization}
    The MLE step in \eqref{prob:kMLE} serves as an update and stopping criterion in KDPE. It does not attempt to approximate the values of the IF. Since we are not initializing an estimation problem from scratch but rather commencing with a pre-existing estimate assumed to be of high quality, there is no necessity in MLE \eqref{prob:kMLE} to conform to any specific functional form or regularity, such as sparsity or other typical assumptions found in nonparametric estimation problems that might justify regularization. The objective of solving \eqref{prob:kMLE} is to update the distributional estimate in a manner that (i) maintains the consistency of the naive plugin (Assumptions \ref{assump:empprocess} and \ref{assump:secondorderremainder}) and (ii) satisfies the score equation \eqref{eq:ScoreEquations} for every function in the RKHS (not a subset). Consequently, we have the freedom to choose any solution for \eqref{prob:kMLE} that fulfills requirements (i) and (ii), without imposing any additional constraints on this selection. 
    
    The proof of the asymptotic linearity and efficiency of KDPE (Theorem~\ref{thm:main}) relies on the universality of the RHKS $\mathcal{H}_{P_n^{\infty}}$ rather than the span of $\{K_{P_n^{\infty}}(O_1, \cdot), ..., K_{P_n^{\infty}}(O_n, \cdot)\}$.
\end{remark}
\subsection{Kernel Debiased Plug-in Estimator (KDPE)} 
Given a pre-estimate $P^0_n$ and a PD kernel $K$, we
\begin{itemize}[itemsep=0pt, parsep=0pt]
  \item[(i)]
    construct fluctuations $\widetilde\calM^0 \equiv \widetilde\Mcal_{P_n^0}$
    around $P^0_n$ defined in
    Equation \eqref{eq:hcalmodel} by computing the zero-mean kernel
    $K_{P^0_n}$ as described in Proposition \ref{thm:mean-zero-kernel};
  \item[(ii)] solve the regularized MLE \eqref{prob:kMLE}
    over scores $\calH_{P^0_n}$ via the equivalent simpler optimization
    \eqref{eq:kMLEsol} to find $\mbf\a^0_n \equiv \mbf\a_n(p_n^0)$,
    which defines the MLE score $[\mbf\a^0_n]^\top k^0_\mbf{O}$ .
\end{itemize} 
As a result, we obtain the $1^\text{st}$-step of our estimator
$ 
  p^1_n
  =
  p_n^0\big([\mbf\a^0_n]^\top k_{\mbf{O}}^0\big)
  =
  \big(1 + [\mbf\a^0_n]^\top k^0_\mbf{O}\big) p^0_n.
$ 
We define the $\ell+1^\text{th}$-step of KDPE recursively by
\begin{align} \label{eq:tmlek} 
  &p^{\ell+1}_n
  \coloneqq
  p_n^\ell\big([\mbf\a^\ell_n]^\top k_{\mbf{O}}^\ell\big)
  =
  \big(1 + [\mbf\a^\ell_n]^\top k^\ell_\mbf{O}\big) p^\ell_n
  \\ &\text{and}\quad
  \mbf\a^\ell_n\equiv \mbf\a_n(p^\ell_n)
  .
\end{align} 
This recursion terminates when $\mbf\a^\ell_n=0$, indicating that
$p_n^\ell$ is a local optimum for the likelihood over submodel
$\widetilde\calM^\ell$.
Similarly to TMLE,
we assume
the following 
about
the convergence (in the iteration step $\ell$) of the KDPE recursion:


\begin{assumption}[Convergence and termination at an interior point] 
  \label{assump:S3}
  We assume that the $\ell^\text{th}$-step of KDPE estimate $p_n^\ell$ converges almost
  surely to some limit $p_n^\infty$ in the interior of $\calM$, and that this
  limiting distribution $P_n^\infty$ is a local minimum of the regularized
  empirical risk \eqref{prob:kMLE} in the interior of the perturbation space
  $\widetilde\calM^\infty$.
\end{assumption} 
Assumption \ref{assump:S3} requires the convergence of our procedure to a limit
$P_n^\infty$, 
ensuring that this limit 
satisfies an interior condition on 1) the full model class $\Mcal$ to ensure that our
problem is regular,\footnote{
  That is, our model is locally fully nonparametric,
  i.e., the tangent space is $L_0^2(P_n^\infty)$,
  and our target parameter is pathwise differentiable, i.e., the conditions in
  Lemma~\ref{lemma:expension} are satisfied.
} and 
2) the submodel $\widetilde\Mcal^\infty$ so that the first-order optimality conditions hold. Proposition
\ref{prop:pnh0} highlights a key consequence of Assumption \ref{assump:S3}:

\begin{prop}\label{prop:pnh0} 
  Under Assumption \ref{assump:S3},
  the final KDPE estimate $P_n^\infty$ of the true distribution $P^*$
  satisfies the following estimating equations:
  \begin{align} \label{eq:ScoreEquations}
    \PP_n h =0 \quad \text{for all } h \in \calH_{P^\infty_n}.
  \end{align}
\end{prop} 

While Proposition \ref{prop:pnh0} states that KDPE solves infinitely many score
equations, it guarantees that the plug-in bias, 
$\PP_n \phi_{P^\infty_n}$,
is eliminated if the influence function $\phi_{P^\infty_n}$ falls in the RKHS $\Hcal_{P^\infty_n}$. By choosing a universal kernel \citep{carmeli2010vector, micchelli2006universal}, we strengthen the Proposition \ref{prop:pnh0} to achieve our
desired debiasing result, under the regularity conditions below.

\begin{assumption} \label{assump:S4}
  We assume that (i) RKHS $\calH_K$ is universal, (ii) the estimate $P^\infty_n$ satisfies 
  $p^*/p^\infty_n \in L^2(P_n^\infty)$, and (iii) the estimated influence function
  $\phi_{P^\infty_n}$ is in $L^2(P^*)$ $P^*$-a.s.
\end{assumption}

\begin{assumption}\label{assump:S5} 
There exist an event $\Omega$ with $P^*(\Omega)=1$,
and for every $\w\in\Omega$ there exists a sequence $h_j(\w)\in \calH_{P^\infty_n(\w)}$ indexed by $\w\in\Omega$ such that the following hold for all $w\in\Omega$: 
(i) 
$
\norm{h_j(\w) - \phi_{P^\infty_n(\w)}}_{L^2(P^*)}
\xrightarrow[j\to\infty]{} 0
$
and (ii) there exists a $j_0(w) \in \NN$, $j_0(w)<\infty$ such that the set of functions
$\Set{h_j(\w)}{\w\in\Omega, j > j_0(w)}$
is $P^*$-Donsker.
\end{assumption} 

\begin{theorem}[Asymptotic linearity and efficiency of KDPE] \label{thm:main} 
Let $\psi:\calM\to\mR$ be a pathwise-differentiable functional of the
distribution $P$ with influence function $\phi_P\in L^2_0(P)$ and von
Mises expansion:
$
  \psi(\bar P) - \psi(P) 
  =
  \int \phi_{\bar P} d(\bar P - P)
  + R_2(\bar P, P)
$
for any $\bar P,P\in\calM$ (which defines the second-order reminder term $R_2(\bar P, P)$.) Under Assumptions \ref{assump:empprocess} (i), \ref{assump:secondorderremainder},
\ref{assump:S3}, \ref{assump:S4}, \ref{assump:S5}, $\mP_n\phi_{P^\infty_n} = o_{P^*}(1/\sqrt{n})$
and the KDPE estimator satisfies
$
  \psi(P^\infty_n) - \psi(P^*)
  =
  \PP_n\phi_{P^*} + o_{P^*}(1/\sqrt{n}).
$
\end{theorem} 
 
Theorem \ref{thm:main} shows
that KDPE attains asymptotic efficiency by converging to the optimal limiting distribution within the class of RAL estimators. 
To converge to this asymptotically optimal distribution, error terms that deviate from this limit distribution must be on the scale of $o(1/\sqrt{n})$. Thus, the definition of asymptotic efficiency (\citealp{van2000asymptotic}, Section 8) also implies that KDPE achieves local asymptotic minimax optimality with respect to the error $\psi(P_n^\infty) - \psi(P^*)$ on the order of $O(1/\sqrt{n})$. The proof of Theorem~\ref{thm:main} (Appendix~\ref{proof:thm_main}) 
%
requires that as $n$ approaches infinity, the $\sqrt{n}$-scaled error between our true value and KDPE estimate converges to the asymptotically normal distribution implied by $\PP_n(\phi_{P^*})$.
The proof consists of three steps: Step 1) 
our uniform approximation guarantees are with respect to the true distribution $P^*$ (as opposed to the KDPE distribution ${P}_n^\infty$), Step 2) decompose the error of the KDPE estimate similar to that of Equation \eqref{eq:asyExpansion}, and Step 3) showing that the empirical process term is of order $o_{P^*}(1/\sqrt{n})$.

\begin{remark}
Assumption \ref{assump:S5} is our novel regularity assumption on the model $\mathcal{M}$ and functional $\psi$ needed to control the plug-in bias term $\PP_n \phi_{P_n^\infty}$ of the proposed KDPE estimator.
Assumptions \ref{assump:empprocess}-\ref{assump:secondorderremainder} and conclusions derived from them are used for the one-step correction and TMLE \cite{van2006targeted, kennedy2023}, and analogous conditions to Assumptions \ref{assump:S3} and \ref{assump:S4} are common for TMLE \cite{van2006targeted}. This is \emph{intentional}:
considering a similar set of assumptions
directly characterizes the additional assumption (Assumption \ref{assump:S5}) needed for KDPE and its practical advantages. Our novel assumption, Assumption \ref{assump:S5}, controls an additional, nonstandard empirical process term that arises in the analysis of the KDPE estimator, as shown in Appendix \ref{append:DGP2}. One trivial condition for Assumption \ref{assump:S5} is that $\phi_{P_n^\infty} \in \mathcal{H}_{P_n^\infty}$ for all sample paths. For estimation problems with nuisances (e.g., densities, conditional regression functions, etc.) over a compact, continuous input space $\mathcal{O}$, sufficient conditions for Assumption \ref{assump:S5} are (i) the choice of the universal kernel (e.g., Gaussian Kernel) and (ii) bounded total variation of the approximating sequence $\{h_j(\omega)\}_{j \in \NN}$ for all $\omega \in \Omega$ (see \citealp{dudley_bv}). While the universal approximation property for our choice of kernel is well-documented, our work is the first to apply this property to reduce first-order bias in nonparametric estimation.
\end{remark}

%% file: simulations.tex
In this section, we 
provide a detailed implementation of KDPE, and evaluate it 
using two simulated DGPs that occur in the causal inference literature, 
validating the results provided by Theorem \ref{thm:main}. 
We note that these models are semi-parametric due to conditional independence constraints. In Appendix \ref{app:simulation}, we provide the necessary modifications to Algorithm \ref{alg:kdpe_main} for our updated distributions to be valid within the model $\Mcal$.
Additional details and results (e.g.,  bootstrap confidence intervals, computational complexity, update mechanism) are included in Appendix \ref{app:simulation}.\footnote{Code: \url{https://github.com/brianc0413/KDPE}.}

\paragraph{Modification for KDPE} To use
solvers, we work with a slightly modified algorithm, Algorithm~\ref{alg:kdpe_main}. The main modifications from the baseline KDPE include the incorporation of a density bound parameter, $c$, and a convergence tolerance function. The constant $c$ avoids violations of condition (ii) of Assumption \ref{assump:S4} (square integrability of $p^*/p_n^\infty$) and enables standard convex/concave optimization software for finding $\alpha^{i+1}$. The convergence tolerance function, $l(\cdot, \cdot)$ acts as a metric tracking the net change in distribution between iterations.\footnote{Because the objective function $L(\cdot)= -\PP_n\log(\cdot)$ is globally convex, when $l(P_n^{i+1}, P_{n}^{i})$ is small, we are close to a solution satisfying the necessary first-order conditions.} For all simulations, we use the convergence tolerance function $l(P,P') = \|p-p'\|_{L_2(\eta)}$, where $\eta$ is the counting measure.
\begin{algorithm}[ht!] 
  \caption{KDPE implementation} \label{alg:kdpe_main}
  \begin{algorithmic}[1] 
  \STATE 
  \textbf{input}: data 
  $\{O_i\}_{i=1}^n$, convergence tolerance
  function $l(\cdot)$, convergence tolerance $\gamma$, density bound  $c$, regularization parameter $\lambda$, pre-estimate $P^0_n$, PD kernel $K$.
  \STATE \textbf{initialize}: convergence tracker $\beta(0) = \infty$ and $\ell=0$.
  \WHILE{$\beta(\ell)  > \gamma$}
      \STATE Construct the mean-zero kernel $K^\ell \coloneqq K_{P_n^\ell}$.  
      \STATE Update the distribution according to the log-likelihood loss function:
      \STATE Set $\alpha_n^\ell$:
	\begin{align*}
	    \quad\quad\mbf{\alpha}^{\ell}_{n} 
	=
	\argmin_{\substack{\alpha \in \RR^n,\\\ c \leq {p}_n^\ell(\alpha^\top k_{\mbf{O}}^\ell) \leq 1-c}}&
	  -\PP_n \log([1+\alpha^\top k_{\mbf{O}}^\ell]p_n^\ell)
   \\ 
        & 
        + \lambda\alpha^\top K_{\mbf{O}}^\ell\alpha
	\end{align*}
	
      \STATE set ${p}_{n}^{\ell+1} = p_n^\ell([\mbf\alpha_n^\ell]^\top k_{\mbf{O}}^\ell) = [1+[\mbf{\alpha}_n^\ell]^\top k_{\mbf{O}}^\ell]p_{n}^\ell$
      \STATE set $\beta(\ell +1) = l(P_n^{\ell+1}, {P}_n^{\ell})$,
	and $\ell = \ell+1$.
  \ENDWHILE
  \STATE \textbf{output}: $\h{P}_{\tsf{KDPE}} = {P}_n^{\ell}$.
\label{alg:modified_KDPE}
  \end{algorithmic}    
\end{algorithm}

\subsection{DGPs and Target Parameters}
We consider two common models in causal inference literature and demonstrate how to use KDPE to debias plug-in estimators for the desired target parameters. For all DGPs, we make the standard
assumptions of 
positivity, conditional ignorability, and the stable unit treatment value assumption (SUTVA),
which enable the identification of our target parameters with a causal interpretation.

\paragraph{DGP1: Observational Study}
We define $\Ocal \equiv \Xcal\times \Acal\times \Ycal$, with baseline covariate $X \in \Xcal \equiv [0,1]$, binary treatment indicator $A \in \Acal \equiv \{0,1\}$, and a binary outcome variable $Y \in \Ycal \equiv \{0,1\}$.  The DGP can be decomposed as follows:
$$P \equiv P_X \times P_{A|X} \times P_{Y|A,X}.$$
We place no restrictions on $P_{Y|A,X}$ other than being a valid conditional density function (i.e., nonparametric). 
The true DGP in our simulation is given below:
\begin{align*}
    & X\sim \text{Unif}[0,1], \quad A|X \sim  \text{Bern}\left(0.5+\frac{1}{3}\sin\left({50X}/{\pi}\right)\right),\\ 
    &Y|A,X \sim \text{Bern}\left(0.4 + A(X-0.3)^2 + \frac{1}{4}\sin({40X}/{\pi})\right).
\end{align*}

\paragraph{DGP2: Longitudinal Observational Study}
A more complicated DGP represents a two-stage study. We define $\Ocal \equiv \Xcal\times \mathcal{A}_0 \times \mathcal{L}_1 \times \mathcal{A}_0\times \Ycal$, with baseline covariate $X \in \RR$, binary treatment indicators $A_0, A_1$ (where $A_i$ is the treatment at time $i$), and binary outcome variables $L_1, Y$. The time ordering of the variables is given by:
$W \rightarrow A_0 \rightarrow L_1 \rightarrow A_1 \rightarrow Y$.
As before, the DGP can be decomposed as follows:
$$P \equiv P_X \times P_{A_0|X} \times P_{L_1|A_0,X} \times P_{A_1|L_1,A_0,X} \times P_{Y|A_1,L_1,A_0,X}$$
which have the following distribution:
\begin{align*}
\centering
     &X \sim \text{Unif}[0,8],\quad  A_0 \sim \text{Bern}(0.5)\\
    &L_1 = \mbf{1}[3+A_0-0.75X + \epsilon_1 > 0]\\
    &A_1 = \mbf{1}[\text{expit}(-3 + 0.5X + 0.4*L1)\geq 0.3]\\
    &Y = \mbf{1}[\text{expit}(X- 3.5 \\
    &\quad\quad\quad-0.3A_0 - 0.5L_1 -0.5A_1 + \epsilon_2)\geq 0.5],
\end{align*}
where $\epsilon_1,\epsilon_2$ are i.i.d. noise with distribution $N(0,1)$. 

\paragraph{Target Parameters}
The target parameters that we consider here are functions of the mean potential outcome under a specific treatment policy. For DGP1, we denote the mean potential outcome as:
$$ \mu_a(P) \equiv P_X[P_{Y|A,X}[Y|A=a, X=x]].$$
For DGP2, we consider the mean potential outcome for a fixed treatment across all time points ($A_1 = A_0 = a$):\begin{align*}
    \mu_a(P) &\equiv P_X\big[P_{L_1|A=a,X}[P_{Y|A_1,L_1,A_0,X}[\\ &\quad\quad\quad Y|A_1=a, L_1 = l_1, A_0=a, X=x]]\big].
\end{align*}
The target parameters considered in this study are:
\begin{itemize}
    \item avg. treatment effect, $\psi_{\textsf{ATE}}(P) = \mu_1(P) - \mu_0(P)$,
    \item relative risk, $\psi_{\textsf{RR}}(P) = \mu_1(P)/ \mu_0(P)$,
    \item odds ratio, $ \psi_{\textsf{OR}}(P) = \frac{\mu_1(P)/(1-\mu_1(P))}{\mu_0(P)/(1-\mu_0(P))}.$
\end{itemize}

\subsection{Simulation Setup}
\paragraph{Pre-estimate Initialization} For both DGPs, we initialize the distribution of baseline covariate $X$ as $P_n^0(X) = \PP_n(X)$ and use the package \texttt{SuperLearner} \citep{sl2007} to estimate the remaining conditional distributions, obtaining $P_n^0$. For our testing of highly adaptive lasso (HAL, \citealp{vanderlaan2017hal}), we use the implementation of HAL provided by package \texttt{hal9001} \cite{coyle2022hal9001-rpkg}.\footnote{We defer additional details on hyperparameter selection and tuning to Appendix \ref{app:simulation}.}

\paragraph{Baseline Methods for Comparison} We compute these estimators with four different plug-in methods: 1) our proposed method KDPE,  2) TMLE, 3) HAL, and 4) a baseline using the biased super-learned distribution, SL \citep{sl2007}, which is our pre-estimate. 
TMLE has access to the functional form of the IF of each target parameter and is known for its efficient limiting distribution.
TMLE 
serves as a benchmark 
for examining the asymptotic efficiency of KDPE.
We expect KDPE to perform no worse than TMLE. HAL serves as a benchmark for the empirical performance of KDPE against an existing method with similar aims: HAL also solves a rich class of score / estimating equations and can be used as a debiased plug-in for many different target parameters.
In contrast, comparisons against SL illustrate the effect of the KDPE approach 
on the distribution of our estimates,  highlighting the need for using a debiasing approach.
For DGP2, we use the longitudinal version of TMLE, denoted as LTMLE \citep{ltmle} as our baseline.  For (L)TMLE, we learn a distinct data-generating distribution for each target parameter, i.e.,
$\h{P}_{\text{TMLE}}^{\text{RR}}, \h{P}_{\text{TMLE}}^{\text{OR}},
\h{P}_{\text{TMLE}}^{\text{ATE}}$, to debias the estimate. KDPE, HAL, and SL only use a single estimated distribution $\h{P}_{\text{SL}}$ and $\h{P}_{\text{KDPE}}$, respectively.

\paragraph{Simulation Settings}
In all experiments, we use 300 samples
($n=300$) and the mean-zero Gaussian kernel. For both DGP1 and DGP2, we set
$c=0.001$. For DGP1, we use 450  simulations, with $\lambda =0, \gamma = 0.002$. For DGP2, we use 350 simulations with $\lambda = 15$, $\gamma = 0.0001$.\footnote{Because larger values of $\lambda$ result in smaller changes, we recommend setting smaller $\gamma$ for larger $\lambda$.} The choice of hyperparameters is discussed in Appendix \ref{app:simulation}.



\subsection{Comparing Empirical Performance}
\label{subsec:emp_dist}
\begin{figure*}[ht!]
    \centering
    \includegraphics[width=0.91\linewidth]{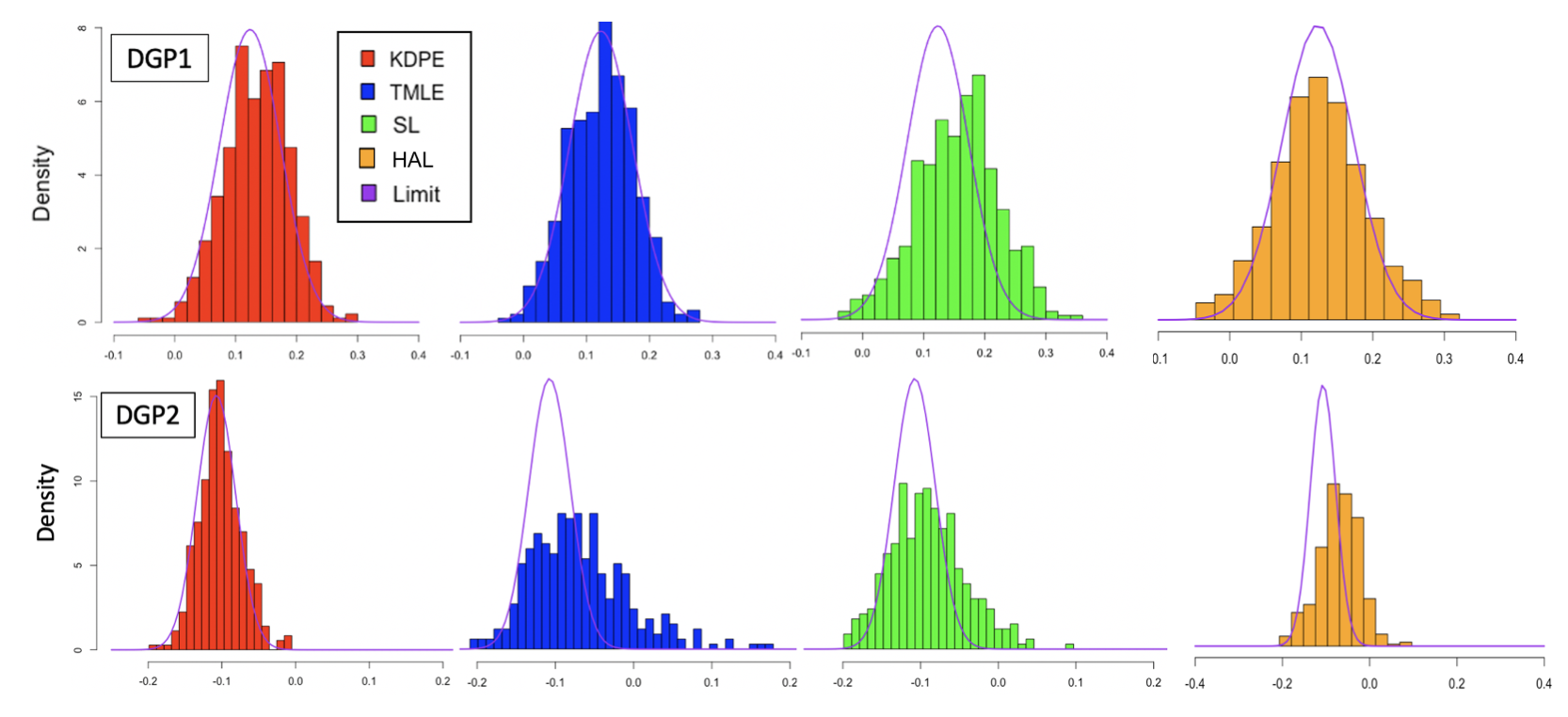}
    \caption{Simulated distributions of $\h \psi_{\mathrm{ATE}}$ compared to their asymptotic distributions. TMLE distribution in the second row corresponds to LTMLE for DGP2.}
    \label{fig:limitdists}
\end{figure*}

\begin{table}[t!]
\centering
\begin{tabular}{rlrrr}
  \hline
  &Method  & $\psi_{\text{ATE}}$ & $\psi_{\text{RR}}$ & $\psi_{\text{OR}}$ \\ 
  \hline
  &SL & 0.0803 & 0.2623 & 0.6796 \\ 
DGP1  &TMLE & 0.0574 & 0.1723 & 0.4059 \\ 
   &KDPE & 0.0592 & 0.1752 & 0.4303 \\ 
   & HAL & 0.0635 & 0.1987 & 0.4795 \\
   \hline
     &SL & 0.0508 & 0.0925 & 0.1555 \\ 
DGP2  &LTMLE & 0.0731 & 0.1481 & 0.2648 \\ 
   &KDPE & 0.0295 & 0.0778 & 0.0827 \\ 
   &HAL & 0.0589 & 0.1084 & 0.1848\\
   \hline
\end{tabular}
\caption{Root Mean Squared Error (RMSE) of KDPE, (L)TMLE, HAL-MLE, and SL for DGP1, DGP2.}
\label{table:RMSE}
\end{table}
\paragraph{DGP1 verifies the results of Theorem \ref{thm:main}} 
Our simulations for DGP1 (Figures~\ref{fig:limitdists} and \ref{figure: asymp_normal}) demonstrate two distinct results of Theorem~\ref{thm:main}: (i) KDPE is roughly equidistant to the efficient limiting normal distribution as TMLE, and (ii) KDPE works well as a plug-in distribution for many pathwise differentiable parameters. 

As an example, the limiting distribution of $\psi_{\mathrm{ATE}}(P^*)$ is given by $N(0.1233, $ $0.0025)$ for DGP1. For DGP1, the distribution of $\psi_{\mathrm{ATE}}(\h{P}_{\text{SL}})$ in Figure \ref{fig:limitdists} (and $\psi_{\mathrm{RR}}(\h{P}_{\mathrm{SL}}), \psi_{\mathrm{OR}}(\h{P}_{\text{SL}})$, shown in Figure \ref{figure: asymp_normal} of Appendix \ref{app:simulation}) 
highlight the necessity of debiasing. Without correction,
the naive plug-in estimator deviates significantly from the limiting distribution, while all three debiasing approaches, $\psi_{\ATE}(\h{P}_{\KDPE})$, $\psi_{\ATE}(\h{P}_{\text{HAL}})$ and $\psi_{\ATE}(\h{P}_{\TMLE}^{\ATE})$, achieve similar distributions to the limiting normal distribution. Despite no knowledge of the influence function when performing debiasing, the simulated distribution of $\psi_{\ATE}(\h{P}_{\KDPE})$ is as close to the efficient limiting distribution as $\psi_{\ATE}(\h{P}_{\TMLE}^{\ATE})$, and closer to the efficient limiting distribution than $\psi_\ATE(\h{P}_{\text{HAL}})$. The root-mean-squared error results provided in Table \ref{table:RMSE} corroborate this result, where the RMSE of KDPE is closer to that of TMLE than HAL. 

The results for DGP1 in Figure \ref{figure: asymp_normal} verifies the second consequence of Theorem \ref{thm:main}: for all pathwise differentiable parameters that satisfy the stated regularity conditions, $\h{P}_{\KDPE}$ is debiased as a plug-in estimator. Across all three target estimands, the KDPE estimate $\psi(\h{P}_{\KDPE})$ performs similarly (both in terms of similarity to the limiting distribution and RMSE) compared with TMLE, which uses a distinct targeted plug-in ($\h{P}_{\TMLE}^{\ATE}, \h{P}_{\TMLE}^{\RRm}, \h{P}_{\TMLE}^{\ORr}$) for each parameter.

\paragraph{DGP2 demonstrates that KDPE may provide improved finite sample performance.} DGP2 provides an example where KDPE vastly outperforms TMLE and HAL for fixed sample size $n=300$, despite TMLE using the influence functions for each estimand. KDPE significantly outperforms TMLE, HAL, and the naive SL plug-in estimates both in terms of RMSE, as shown in Table \ref{table:RMSE} and convergence to the asymptotic limiting distribution, as shown in Figure \ref{fig:limitdists} for $\psi_\ATE$. The resulting plug-in estimates from KDPE demonstrate much smaller variance (Figures \ref{fig:limitdists} and \ref{figure: asymp_normal}) than SL, HAL, or TMLE for all target parameters, and is centered correctly for $\psi_{\mathrm{ATE}}, \psi_{\mathrm{OR}}$. The performance of KDPE for $n=300$ indicates that for smaller sample sizes, KDPE may outperform a more targeted approach; in contrast, the IF-based approach, LTMLE, requires a much larger sample size ($n \approx 800$) to reach the correct limiting distribution.

%% file: discussion.tex
In this paper, we introduce KDPE, a new method for debiasing plug-in estimators that (i) does not require the IF as input, and (ii) works as a general plug-in for any target parameters that satisfy our regularity conditions. 

To see the practical benefits of KDPE, let us compare the required work from the analyst (1) using existing IF-based methods, (2) using existing IF-free methods, and (3) using KDPE. An analyst may initially want to estimate $k$ parameters from the distribution, $\{\psi_i(P^*)\}_{i=1}^k$ from the data. Using IF-based methods (DML, TMLE, one-step), the analyst is required to derive (or find) the IF for each target parameter $\psi_i$, and uses the IF to debias each parameter, resulting in $k$ analytical derivations and $k$ runs of the debiasing algorithm. Using finite-differencing approaches, 
the analyst approximates 
the plug-in bias term $\frac{1}{n}\sum_{i=1}^n \phi_{\hat{P}}(O_i)$, and uses the one-step correction,
resulting in $k$ runs of the finite differencing algorithm. Using KDPE, the analyst simply obtains a distributional estimate $\hat{P}$ from running KDPE \textbf{once}, and uses this as a plug-in estimator for the $k$ parameters, resulting in estimates $[\psi_i(\hat{P})]_{i=1}^k$.  After estimating the $k$ parameters, the analyst wishes to do a follow-up analysis by estimating an additional $w$ parameters, indexed $[\psi_i(P^*)]_{i=k+1}^{k+w}$. 
Using IF-based or finite-differencing approaches, the analyst must derive or approximate the IF $w$ additional times, and perform $w$ additional runs of their debiasing algorithm. Using KDPE, we simply use the same $\hat{P}$ obtained previously as a plug-in to obtain estimates $[\psi_i(\hat{P})]_{i=k+1}^{k+w}$. 

This example demonstrates the benefits of KDPE; it (1) avoids any computation/derivation of the IF for all $k+w$ parameters, and (2) obtains a single debiased plug-in distribution $\hat{P}$ that can be used to obtain debiasing for many target parameters. In contrast with existing approaches, KDPE only needs to be run once to obtain a general plug-in distribution for all $k+w$ parameters, and greatly simplifies the estimation process for many different parameters of interest.

Future work includes
(i) analyzing the effect of RKHS-norm regularization on the convergence and efficiency, (ii) finding simple heuristics for setting hyperparameters $c, \gamma, \lambda$, 
(iii) investigating how to apply sample-splitting and whether it can relax our assumptions, and (iv) verifying the validity of bootstrap.
\section*{Acknowledgements}
This work was supported by the NDSEG fellowship and Harvard Data Science Initiative Fellowship. Additionally, this research was completed while Kyra Gan was a postdoc at Harvard, and was supported by NIH grants P41EB028242 and P50DA054039. The authors would like to thank Susan Murphy for her early feedback.
\section*{Impact Statement}
This paper presents work whose goal is to advance the field of Machine Learning. There are many potential societal consequences of our work, none which we feel must be specifically highlighted here.

%% file: notation.tex

\begin{description}[
  leftmargin=2.2cm,
  labelwidth=2cm,
  labelsep=.2cm,
  parsep=0mm,
  itemsep=2pt
  ]

\item[${[}n{]}$] set of integers $1,\ldots, n$

\item[$\calO$] Sample space, with $\calO\subset\mR^d$ and Borel $\s$-algebra $\calB$;

\item[$\calM$] Statistical model, collection of probability measures on $(\calO,\calB)$;

\item[$\psi(P)$] Target parameter functional, mapping from $\Mcal \rightarrow \RR$

\item[$\phi_P$] Influence function, mapping from $\Ocal \rightarrow \RR$

\item[$Pf, P{[}f{]}$] Expectation functional $Pf = P[f] = \int f \, dP$

\item[$\mP_n, \mP_nf$] Empirical distribution function

\item[$L^2_0(P)$] The set of measurable functions $\{f:\Ocal \rightarrow \RR: Pf = 0 \ , Pf^2 < \infty\}$
    
\item[$P^*$] (Unknown) data generating distribution;
\item[$o_{P^*}(1/\sqrt{n})$] Represents a sequence of random variables that converges to zero in probability with respect to probability measure $P^*$, at a rate that is faster than $1/\sqrt{n}$.
\item[$K, \Hcal_K$] The PSD kernel and associated RKHS with kernel function $K(\cdot, \cdot)$

\item[$k_x$] The function 
  $k_x \coloneqq K(x, \cdot): \Xcal \rightarrow \RR$, for $x$ in support $\Xcal$

\item[$k_\mbf{x}$] Vectorized version of function $k_x$, where $k_\mbf{x}\coloneqq [k_{x_1},...,k_{x_n}]^T$

\item[$K_\mbf{xy}, K_\mbf{x}$] Matrix of values $K_\mbf{x} \coloneqq [K(x_i, x_j)]_{i,j} \in \RR{n\times n}$, $K_{\mbf{xy}} \coloneqq [K(x_i, y_j)]_{ij}\in \RR{n\times m}$

\item[$K_P, \Hcal_P$] The mean-zero kernel (with respect to $P$) and associated RKHS

\item[$\calC_0(\calX)$] The space of continuous functions that vanish at
infinity with the uniform norm

\item[$\calC_0$] A kernel $K$ is $\calC_0$ provided that $\calH_K$ is a subspace of
$C_0(\calX)$

\item[$\widetilde{\Mcal}_P$] The collection of perturbations $\widetilde{\Mcal}_P\coloneqq\{(1+h)p \ |\ h \in \Hcal_P \text{ such that }(1+h)p > 0\}\cap \Mcal$

\item[$\widetilde{\Mcal}_{P, n}$] The set of $n$-dimensional parametric submodels, defined as follows:
\[\widetilde{\Mcal}_{P, n} \coloneqq\Set[\big]{
    p_{\bm{\a}} \coloneqq p(\bm{\a}) \coloneqq
    (1 + \bm{\a}^\top k_\mbf{O})p
  }{
    \bm{\a}\in\mR^n \text{ such that } p(\bm{\a}) > 0
  } \cap \calM \]

\item[$\widetilde\calM^\ell$] The collection of perturbations $\widetilde\Mcal_{P_n^l}$ corresponding to $\ell$-th iteration's distribution, $P_n^l$

\item[$p(\epsilon|h),\ p_{\epsilon, h}$] 1-dimensional parametric model $(1+\epsilon h)p$

\item[$P_{\epsilon,h}$] The distribution corresponding to the density implied by the 1-dimensional parametric model $(1+\epsilon h)p$

\item[$p(h), p_h$] RKHS-based model $(1+h)p$, where $h \in \Hcal_{P}$


\item[$L_{p, h}$] Log-likelihood likelihood loss, given by $L_{p,h}\equiv -\log(p_h)$



\item[$D_P\psi$] Derivative functional (assumed to be continuous, linear) that maps $L_0^2(P)\rightarrow \RR$,\newline $\frac{d}{d\epsilon}\psi({P_{\epsilon, h}})|_{\epsilon=0}=D_P\psi[h]$ for all $h \in L_0^2(P)$

\item[$\pi_S$] The orthogonal projection operator $\Hcal \rightarrow S$, projects a function $h \in \Hcal$ onto a set of functions $S$

\item[$\bm\a_n^i$] The optimal perturbation for objective function $L(\h{p}_n(\alpha, p), \lambda)$ at iteration $i$ for KDPE

\item[$P_n^i, \ p_n^i$] Estimated data distribution of TMLE/KDPE at iteration $i$

\item[$P^\infty_n, \ p_n^\infty$] Limiting/final estimate for TMLE/KDPE

\item[$\h{P}$] Generic estimated distribution, subscripts/superscript added accordingly

\item[$\h\psi_{\text{param.}}$] Estimated target parameter value $\psi_{\text{param.}}(\h P)$ using plug-in distribution $\h P$, subscripts/superscript added accordingly

\item[$Q(P)$] Relevant portion of distribution (e.g. $
\{Q_a(X)\}_{a \in \{0,1\}}= \{P[Y=1|X, A=a]\}_{a\in \{0,1\}}$).

\item[$\Omega$] An event (or set of events) such that has probability 1 under $P^*$, i.e. $P^*(\Omega)=1$. The standard example is the set of all possible $n$-sample realizations $\{\{O_i(w)\}_{i=1}^n \ :\ w \in \Omega\}$ from distribution $P^*$.   

\item[$f(w), P(w)$] The random function / distribution, indexed by events $w \in \Omega$. Often, we use $w \in \Omega$ to index an $n$-sample realization, $\{O_i(w)\}_{i=1}^n$, from DGP $P^*$.   
\end{description}

%% file: appendixB.tex


\subsection{Additional Details Regarding $P_{\epsilon, h}$ and  $\widetilde{\Mcal}_P$.} \label{app:questions_from_GgGE}

\paragraph{Validity of Linear Submodels.} The linear submodel $P_{\epsilon, h}$ is indeed a valid submodel (i) for all $h \in L_0^2(P)$ and (ii) for $|\epsilon| \leq \epsilon_{h, P}$, where $\epsilon_{h, P}$ is an upper bound on the magnitude of $\epsilon$ depending on $h$ and $P$.  For example, $h = p'/p - 1$, the path connecting $P$ and another unique (up to the dominating measure) distribution $P' \in \mathcal{M}$, has $0 \leq \epsilon \leq \epsilon_{h,P} = 1$, which corresponds to $p_{\epsilon,h} = (1-\epsilon)p + \epsilon p'$; for $h' = h/2$, our bound $\epsilon_{h', P}$ becomes 2. Because we are in the nonparametric setting, the only constraint on our model class $\Mcal$ is that we remain a valid density, and thus our linear paths $P_{\epsilon, h}$ are fully contained in the model class $\Mcal$.

\paragraph{Additional Comments Regarding $\widetilde\Mcal_P$.} We note that $\widetilde{\Mcal}_P$ will never be empty in the nonparametric case. Intuitively, Eq. \ref{eq:hcalmodel} simply restricts the feasible scores $h \in L_0^2(P)$ to a subset of scores $h \in \Hcal_P$ , which is dense in $L_0^2(P)$ and therefore nonempty for any $P$. At each $P$, the nonparametric tangent space is $L_0^2(P)$, which means that all $h \in L_0^2(P)$ are feasible directions. In other words, one can move at least an infinitesimal amount in each $h \in L_0^2(P)$ and remain within the model class. In our bounded kernel setting with universal (e.g. Gaussian) kernels, $\Hcal_P$ is a dense subset of 
$L_0^2(P)$. This implies for each $P$ and $h \in \Hcal_P$, $p(h)$ is a valid submodel in , given a correct magnitude (i.e. obeys $\epsilon_{h,P}$ bounds implied by $\Mcal$, where $h$ is rescaled to be unit norm). Thus, the restriction imposed in Eq. \ref{eq:hcalmodel} corresponds to bounding the maximum magnitude $\epsilon$ we can move for each unit-norm $h \in \Hcal_P$ (i.e. $\epsilon_{h, P}$), such that we remain a valid density.

\subsection{Proof of Proposition~\ref{thm:mean-zero-kernel}}
\label{proof:mean-zero-kernel}
\begin{proof}[Proof of Proposition~\ref{thm:mean-zero-kernel}]
Because $K$ is bounded, the expectation functional $h\mapsto Ph$ is continuous
on $\calH_K$.
Next, we note that $\varphi \in \calH_K$ and it is the Riesz representer of the expectation functional (we have shown this elsewhere and remark here that the rigorous argument considers a Riemann approximation to the expectation integral and weakly convergent subsequences).
It follows that $h\in\calH_P$ if and only if $\brk{h}{\varphi}_K = 0$, and from this we can deduce the projection operator onto $\calH_P$ and the reproducing kernel of that subspace.

Finally, given $h\in L^2_0(P)$, let $\set{g_n}\subset \calH_K$ be an approximating sequence of $h$ in $\calH_k$.
Let $g_n =\tilde g_n + \mE_P[g_n]\varphi/\mE_{P\times P}[K]$ be the RKHS orthogonal decomposition of $g_n$ into $\calH_P$ and $\calH_P^\perp$, and note that this operation is continuous on $L^2(P)$.
By continuity of expectation in $L^2(P)$, $\mE_P[g_n] \to 0$ in $L^2(P)$.
It follows that the orthogonal part of $g_n$ vanishes in $L^2(P)$ and $\tilde g_n \to h$ in $L^2(P)$.
Thus,
$\calH_P$ is dense in $L^2_0(P)$.
\end{proof}

\subsection{Proof of Theorem~\ref{thm:representer}}
\label{proof:thm_representer}

Let $L(Q(P)):\calX\to\mR$ be a general loss function that identifies the relevant part, $Q(P)$, for the target parameter of interest $\psi(P) = \psi(Q(P))$, satisfying the property 
that $\mE_{P} L(Q(P)) \le \mE_{P} L(Q')$ for every possible $Q'$.
The log-likelihood $L(P) = -\log p$ is a special case of such a loss function. 
The risk associated with the log-likelihood
identifies the entire distribution $P$ based on the properties of the Kullback–Leibler divergence.


\begin{proof}[Proof of Theorem~\ref{thm:representer}]
For $h\in\calH$ consider the orthogonal decomposition $h = h_S + h_S^\perp$ into the closed subspaces $S_n$ and $S_n^\perp$.
From the reproducing property it follows that 
$h(O_i) = \brk{h_S + h_S^\perp}{K_{O_i}} = h_S(O_i)$ and $h_S^\perp(O_i)=0$ for every $i\in [n]$.
This means that inputs $h$ and $h_S$ yield the same value of the empirical risk $\sum_i \tilde L(O_i)$ in the objective function of \eqref{prob:kMLE}.
Furthermore, we have assumed that if $h$ is admissible in the optimization problem \eqref{prob:kMLE}, then so is $h_S$.
The claim follows by noting that 
$\norm{h}^2_\calH = \norm{h_S}^2_\calH +\norm{h_S^\perp}^2_\calH$.
\end{proof}

\subsection{Proof of Proposition~\ref{prop:pnh0}}
\label{proof:prop_pnh0}
\begin{remark}\label{remark:prop_pnh0}
Assumption~\ref{assump:S3}
assumes that our algorithm converges and terminates in finite time. 
It is important to highlight that this assumption
    is commonly adopted in the TMLE literature~\citep{van2006targeted}. As demonstrated in Table~\ref{table:hyperparamtuning},  our algorithm, Algorithm~\ref{alg:modified_KDPE}, converges effectively in practice, typically requiring only a few iterations. While the regularization of the perturbation size (i.e., 
    the norm of $h$) could potentially ensure finite convergence, it is beyond the scope of this paper and remains a topic for future research.
\end{remark}
\begin{proof}[Proof of Proposition~\ref{prop:pnh0}]
This follows from the first order condition to \eqref{prob:kMLE} which must hold at the interior local minimum.
First, consider the variations at $p^\infty_n$ in $\widetilde\calM_{P^\infty_n}$ along the $n$ directions (scores) given by $K_{P^\infty_n}({O_1}, \cdot), \ldots, K_{P^\infty_n}({O_n}, \cdot)$ and observe:
\begin{align*}
    0
    =&
    \frac{\partial}{\partial_{\a_j}}\Big|_{\a=0} \sum_i 
    -\log\Big( 1+\sum_j\a_j (K_{P^\infty_n})(O_j,O_i) \Big) p^\infty_n(O_i)
    +
    \lambda_K \a^\top K_\mathbf{O}\a
    \\
    &\qquad =
    \left[ -\sum_i \frac{
      K_{P^\infty_n}(O_j,O_i) p^\infty_n(O_i)
    }{
      (1+\sum_j\a_j (K_{P^\infty_n})(O_j,O_i)) p^\infty_n(O_i)
    }
    +
    \lambda_K (K_\mathbf{O}\a)_j \right]\Bigg|_{\a = 0}
    \\
    &\qquad =
    -\sum_i K_{P^\infty_n}(O_j,O_i) .
\end{align*}
This shows that the estimating equation holds for each of these scores, and therefore their linear combinations.
Second, by Theorem \ref{thm:representer}, any $h\in\calH_{P^\infty_n}$ that is orthogonal to these scores must have $h(O_i)=0$ for all $i\in [n]$ and therefore satisfy 
Equation \eqref{eq:ScoreEquations}.
From these two observations, we conclude the proposition.
\end{proof}

\subsection{Proof of Theorem~\ref{thm:main}}
\label{proof:thm_main}

\begin{remark}\label{remark:thm_main}
Assumptions \ref{assump:secondorderremainder}-\ref{assump:S3} and conclusions derived from them are standard in the TMLE literature \citep{van2006targeted, kennedy2023}.
Assumption \ref{assump:S4} is mostly technical and ensures that the set of scores that satisfy the estimating equation \eqref{eq:ScoreEquations} is sufficiently rich to achieve debiasing for the target parameter by approximating its efficient influence curve estimating equation asymptotically.

Assumption \ref{assump:S5} is our main regularity assumption on the model $\calM$ and functional $\psi$ needed to control the plug-in bias term $\mP_n \phi_{\h P}$ of the proposed KDPE estimator.
While this assumption provides a clear intuition for the debiasing mechanism of KDPE, it is not strictly required for the method to work.
Less extreme conditions on
the regularity of $\calF$, such as $C^\a$ or Sobolev norm bounds (see \citealt[Ch 19]{van2000asymptotic} and RKHS approximation theory \citealt{smale2007learning}), can imply the Donsker property in \textbf{S\ref{assump:S5}}.

The Donsker assumption \textbf{S\ref{assump:secondorderremainder}} has been successfully relaxed for related estimators in the literature \citep{zheng2010, kennedy2023, chernozhukov2017double} via sample-splitting.
Our initial analysis of DKPE crucially relies on the Donsker assumption in the argument for asymptotic debiasing (based on \textbf{S\ref{assump:S5}}).
We leave the analysis of KDPE with sample-splitting (as done in \citealt{zheng2010}) to future work.
\end{remark}





Recall that the proof consists of three steps: Step 1)  
our uniform approximation guarantees are with respect to the true distribution $P^*$ (as opposed to the KDPE distribution ${P}_n^\infty$),
and Step 2) decompose the error of the KDPE estimate similar to that of Equation (2), and Step 3) showing that the empirical process term is of order $o_{P^*}(1/\sqrt{n})$.  

\textbf{Step 1} is achieved by establishing the existence of a sequence in $\mathcal{H}_{P_n^\infty}$ that approximates $\phi_{P_n^\infty}$ arbitrarily well with respect to the $L^2(P^*)$ norm via change of measure
(via Assumption 4, square integrability of density ratios, and the density of $\mathcal{H}_{P_n^\infty}$ in $L_0^2(P_n^\infty)$.)
In \textbf{Step 2}, by using the result from Step 1, the gradient condition of Assumption 3, and the expansion in Lemma 1, we obtain that the error of the KDPE estimate can be decomposed as four terms: (i) the empirical average of $\phi_{P^*}$, which gives us our desired result, (ii) a standard empirical process term across other estimation methods (e.g. TMLE, DML, one-step), (iii) a standard second-order remainder, and (iv) an empirical process term that involves the approximating sequence $h_{j}$. Terms (ii) and (iii) are controlled by Assumptions 1 and 2 respectively such that they are $o_{P^*}(1/\sqrt{n})$. 
In \textbf{Step 3}, we use our novel assumption, Assumption 5, to 
establish that the nonstandard term (iv) such that it is of order $o_{P^*}(1/\sqrt{n})$. Thus, terms (ii-iv) vanish at the appropriate rate, resulting in the conclusion of Theorem 2.

\begin{proof}[Proof of Theorem~\ref{thm:main}]

Because $\phi_{P_n^\infty} \in L_0^2(P_n^\infty)$, and the RKHS $\Hcal_{P_n^\infty}$ is dense in $L^2_0(P_n^\infty)$, there must exist a sequence $\{h_j\}_{j \in \NN} \subset \Hcal_{P_n^\infty}$, where for all $\epsilon > 0$, there exists a $j(\epsilon) \in \NN$ such that for all $j > j(\epsilon)$, $\|h_j - \phi_{P_n^\infty} \|_{L^2(P_n^\infty)} < \epsilon$.

To get the desired convergence rates with respect to true measure $P^*$, we first note that under Assumption \ref{assump:S4}, $|P_n^\infty
(p^*/p_n^\infty)|$ is bounded by some constant $\eta(\omega)$ for sample paths $\omega \in \Omega$ by Holder's inequality:
\[P_n^\infty(p^*/p_n^\infty) \leq P_n^\infty[1]^{1/2} \times \underbrace{\| p^*/p_n^\infty\|_{L^2(P_n^\infty)}}_{< \infty} = \eta(\omega) \]
Assumption \ref{assump:S4} also enables change of measure by a dominating measure $\lambda$. By the Cauchy-Schwartz inequality, 
\begin{align*}
    \|h_j - \phi_{P_n^\infty}\|_{L^2(P^*)}^2 &= \int (h_j - \phi_{P_n^\infty})^2\ dP^* 
    = \int (h_j - \phi_{P_n^\infty})^2 (\frac{p_n^\infty}{p_n^\infty}) p^* d\lambda(O)\\
    &\leq \int \frac{p^*}{p_n^\infty} dP_n^\infty \times\int(h_j - \phi_{P_n^\infty})^2 dP_n^\infty(O) \\
    &\leq \eta(\omega) \|h_j-\phi_{P_n^\infty}\|_{L^2(P_n^\infty)}^2 
\end{align*}

By setting $j > j(\epsilon/\sqrt{\eta(w)})$, this leads to an equivalent result under $L^2(P^*)$ norm (rather than $L^2(P^\infty_n)$ norm):
\[
\|h_j - \phi_{P_n^\infty}\|_{L^2(P^*)}\leq \sqrt{\eta(\omega)}\|h_j - \phi_{P_n^\infty}\|_{L^2(P_n^\infty)} < \epsilon.
\]

Thus, for all $\epsilon > 0$, there exists $j(\epsilon/\sqrt{\eta(\omega)})\in \NN$ such that $j > j(\epsilon/\sqrt{\eta(\omega)})$ implies $\|h_j - \phi_{P_n^\infty}\|_{L^2(P^*)} < \epsilon$. An immediate consequence is another useful result for our analysis. By a simple application of Holder's inequality, 
\begin{align*}
|P^*(h_j - \phi_{P_n^\infty})| &\leq 
P^*[1 \times |h_j - \phi_{P_n^\infty}|]\\
&\leq P^*(|h_j-\phi_{P_n^\infty}|^2)^{1/2} \times P^*(1^2)^{1/2} \quad\quad \text{(Holder's Inequality for $p=q=2$.)}\\
&= \|h_j - \phi_{P_n^\infty} \|_{L^2(P^*)} < \epsilon.
\end{align*}

These two results show that for any sample path $\omega \in \Omega$, there exists a sequence $\{h_j\}_{P_n^\infty} \subset \Hcal_{P_n^\infty}$, such that for all $\epsilon >0$, there exists a $j(\epsilon/\sqrt{\eta(\omega)}) \in \NN$ where $j > j(\epsilon/\sqrt{\eta(\omega)})$ implies the following:
\[
\textbf{(A)}:\ \|h_j - \phi_{P_n^\infty}\|_{L^2(P^*)}< \epsilon,
\quad\quad
\textbf{(B)}:\ |P^*(h_j - \phi_{P_n^\infty})| < \epsilon.
\]

Under the conditions of Assumption \ref{assump:S3} and \ref{assump:S4}, we analyze our estimator using the expansion in Lemma \ref{lemma:expension}:
\begin{align}
    \psi(P_n^\infty) - \psi(P^*) = \PP_n\phi_{P^*} \underbrace{- \PP_n\phi_{P_n^\infty} + (\PP_n - P^*)\big[\phi_{P_n^\infty} - \phi_{P^*}\big]}_{(a)=\text{``double empirical process term''}}+ \underbrace{R_2(P_n^\infty, P^*)}_{(b)=\text{second-order remainder}}.\label{eq:ate_theory_expansion}
\end{align}
Under Assumption \ref{assump:secondorderremainder}, term $(b)$ is $o_{P^*}(1/\sqrt{n})$, and is therefore asymptotically negligible. It remains to show that term $(a)$ vanishes appropriately. Let $\{j^*(\omega, n)\}_{n\in\NN}$ be the index sequence such that $\epsilon = o_{P^*}(1/\sqrt{n})$ for sample path $\omega$. Under Assumption \ref{assump:S3}, $\PP_n h =0$ for every $h \in \Hcal_{P_n^\infty}$, giving us the following equality:
\[(a)= - \PP_n(\phi_{P_n^\infty})+ (\PP_n-P^*)(\phi_{P_n^\infty}-\phi_{P^*})= - \PP_n({-h_{j^*}}+\phi_{P_n^\infty})+ (\PP_n-P^*)(\phi_{P_n^\infty}-\phi_{P^*}).\]

Given property \textbf{(B)}, we can rewrite the "double empirical process term" $(a)$ as follows:
\begin{align*}
    (a) &= -\PP_n(\phi_{P_n^\infty}-h_{j^*}) + (\PP_n-P^*)(\phi_{P_n^\infty}-\phi_{P^*}) \\
  &= -\PP_n(\phi_{P_n^\infty}-h_{j^*}) + (\PP_n-P^*)(\phi_{P_n^\infty}-\phi_{P^*}) + P^*(\phi_{P_n^\infty} - h_{j^*}) \underbrace{-P^*(\phi_{P_n^\infty} - h_{j^*})}_{= o_{P^*}(1/\sqrt{n})}\\
    &=-[\PP_n(\phi_{P_n^\infty})-P^*(\phi_{P_n^\infty})] + [\PP_n(h_{j^*})-P^*(h_{j^*})] + (\PP_n-P^*)(\phi_{P_n^\infty}-\phi_{P^*})  + o_{P^*}(1/\sqrt{n})\\
    &=(\PP_n-P^*)(h_{j^*}-\phi_{P^*}) +  o_{P^*}(1/\sqrt{n})\\
    &= \underbrace{(\PP_n-P^*)(h_{j^*}-\phi_{P_n^\infty})}_{=(c)} + \underbrace{(\PP_n - P^*)(\phi_{P_n^\infty} - \phi_{P^*})}_{=(d)} +  o_{P^*}(1/\sqrt{n}).\label{eq:(a)}
\end{align*}
Term $(d) = o_{P^*}(1/\sqrt{n})$ by Assumption \ref{assump:empprocess}, and term $(c) = o_{P^*}(1/\sqrt{n})$ by Assumption \ref{assump:S5} and Property (\textbf{A}) via Lemma 19.24 of \citealp{van2000asymptotic}. Thus, term $(a) = o_{P^*}(1/\sqrt{n})$. We conclude that $\psi_{P_n^\infty} - \psi(P^*) = \PP_n\phi_{P^*} + o_{P^*}(1/\sqrt{n})$.

\end{proof}

%% file: appendixD.tex
This section provides pseudocode for the bootstrap procedure, empirical justification for our choice of hyperparameters $c, \gamma$, and additional experiments for the bootstrap variance estimation procedure. All simulations were performed on a Dell Desktop with 11th Gen Intel Core i7, 16 GB of RAM. All code for the simulations provided in both this section and Section \ref{sec:sims} can be found on 
https://github.com/anonymous/KDPE. 

\subsection{Baseline Methods}\label{app:sim_mean_zero_construction}
\paragraph{Initializing $P_n^0$ for KDPE} For all simulations, we use the \texttt{SuperLearner} \citep{sl2007} and \texttt{tmle, ltmle} \citep{tmle_package, ltmle} packages in \texttt{R}. For \texttt{SuperLearner}, we set our initial base learners as \texttt{SL.randomForest}, \texttt{SL.glm}, and \texttt{SL.mean}. For \texttt{tmle, ltmle}, we use the default settings, except for setting the parameter \texttt{g.bound} to 0. For each run of the simulations, we initialize $P_0^n$ for both TMLE and KDPE to the same initial density estimate to capture the differences in the two debiasing approaches. 

\paragraph{Optimization Step Details} For the optimization solver, we use the package \texttt{CVXR} \citep{cvxr2020} with the default settings to obtain $\alpha^{i+1}$, which corresponds to  the parametric update step in Algorithm~\ref{alg:modified_KDPE}. Because the log-likelihood loss function $L(\alpha, P)= \PP_n\log((1+\sum_{i=1}^n \alpha_i K_P(o_i, \cdot))p)$ is concave and all constraints are linear, our update step can be solved by existing convex optimization software. 
To avoid solvers, a potential future direction is to explore alternative loss functions that have known closed-form solutions, such as the mean squared error (MSE), rather than using the log-likelihood loss.

\paragraph{Construction of $K_P(o_i, \cdot)$} To construct our $n$-dimensional parametric submodels, we use the mean-zero RBF kernel $K_P(\cdot, \cdot): \Ocal \times \Ocal \rightarrow \RR$, which takes the following form:
\[
K_P(O,O') = \exp(-\|O-O'\|^2_2) - \frac{\int \exp(-\|O-s\|^2_2)dP(s) \int \exp(-\|O'-s\|^2_2)dP(s)}{\int\int \exp(-\|s-t\|^2_2)dP(s)dP(t)},
\]

where $O, O' \in [0,1]\times\{0,1\}^2$. Following classical examples \citep{book2018} and existing software \citep{tmle_package} for TMLE, we fix $X \sim \PP_n(x)$ and $A \sim P_n^0(A|X)$ (the initial estimate for the propensity score function). This enables us to directly calculate the integral, rather than approximate $K_P$:
\begin{equation}\label{eq:kp_experiments}
    K_P(O,O') = \exp(-\|O-O'\|^2_2) - \frac{f_P(O)f_P(O')}{c_P},
\end{equation}
where $f_P$ and $c_P$ are defined as follows:
\begin{align*}
    f_P(O') &= f_P((x',a',y')) 
    = \frac{1}{n}\sum_{i=1}^n\sum_{(a,y) \in \{0,1\}^2} P(A=a,Y=y|X=x_i)\exp(-\|(x',a',y')-(x_i,a,y) \|^2),\\
    c_P &= \frac{1}{n}\sum_{i=1}^n\sum_{a \in \{0,1\}}\sum_{y \in \{0,1\}} P(A=a, Y=y|X=x_i)f(x_i,a,y).
\end{align*}
Note that this calculation is inherent to our setting of binary treatments $A$ and outcomes $Y$ with the empirical marginal distribution for $X$ -- other set-ups may not permit a direct calculation of the integral, and may require integral approximations. We defer those experiments to future work. 

\paragraph{Settings for HAL-MLE} We use the HAL implementation provided by \citealp{coyle2022hal9001-rpkg}. For DGP1, we estimate the conditional densities $P(Y|A,X)$; for DGP2, we estimate conditional densities $P(Y|A_1, L_1, A_0, X)$ and $P(L_1|A_0,X)$. All estimation with HAL is done with base settings, using the link function setting \texttt{family = "binomial"} to respect the $[0,1]$ bounds we make on our data-generating process.

\subsection{Restricted Tangent Space Modifications}\label{app:DGP12} 
The following update steps for DGP1 and DGP2 provide the projection step required for the restricted tangent spaces tested in our simulations. These updates provide explicit examples of conditional tangent space updates, where one is only allowed to update a specific component of the initial DGP estimate, but imposes non-parametric assumptions on this component. A simple example of this restricted model setting is treatment effect estimation with known propensity scores, but no restrictions on the conditional outcome function. Another example is offline estimation of a proposed policy's value, where the treatment assignment policy is specified by the user, but no restrictions are made on the the reward distribution. In the exposition of KDPE, we assume that we update the entire density $P_n^\ell(O)$, and therefore $L_0^2(P_n^\ell)$ is the correct tangent space under our nonparametric assumptions. However, in our simulations, we only aim to perturb the conditional density functions (i.e. $P(Y|A,X)$ in DGP1, $P(L_1|A_0, X), P(Y|A_1, L_1, A_0, X)$ in DGP2). For example, in DGP1, the scores (i.e. perturbation directions) of the conditional density function $P(Y|A,X)$ live in the conditional/projected tangent space $L_0^2(P_n^\ell(Y|A,X))$, as opposed to the more general tangent space of $L_0^2(P_n^\ell)$.

To get the specific updates for DGP1, DGP2, we project the chosen score (i.e., $h \equiv [\alpha_n^\ell]^\top k_{\mbf{O}}^\ell$) into the relevant tangent space. The only change in Algorithm \ref{alg:kdpe_main}
is the following step, which occurs at the end of each loop:
$${p}_{n}^{\ell+1} = p_n^\ell([\mbf\alpha_n^\ell]^\top k_{\mbf{O}}^\ell) = [1+[\mbf{\alpha}_n^\ell]^\top k_{\mbf{O}}^\ell]p_{n}^\ell.$$

For all set-ups in our simulation, the tangent space factorizes nicely, which means we simply need to project our chosen score $h$ at each iteration. We demonstrate the necessary modifications to obtain our results. 

\paragraph{DGP1} Since the conditional density function $P_n^\ell(A|X)$ does not show up in the expression of our final plug-in estimate of the target parameters (ATE, RR, OR), 
it suffices to consider only updating the conditional density function $P_n^\ell(Y|A,X)$, with distribution $P_n^\ell(X) = \PP_n(X)$ fixed. The tangent space $L_0^2(P_n^\ell)$ decomposes as follows:

\[L_0^2(P_n^\ell) = L_0^2(P_n^\ell(X)) \oplus L_0^2(P_n^\ell(A|X)) \oplus L_0^2(P_n^\ell(Y|A,X)), \]

where $L_0^2(P_n^\ell(Y|A,X))$ is space of scores (i.e., directions of change) corresponding to $P_n^\ell(Y|A,X)$. To obtain the component of our chosen score $h = [\alpha_n^\ell]^\top k_{\mbf{O}}^\ell$ relevant for $P_n^\ell(Y|A,X)$, we simply project this into $L_0^2(P_n^\ell(Y|A,X))$:
\[h' = h - P_n^\ell(h|A,X), \quad\quad p^{\ell+1}(Y|A,X) = (1+h')p^\ell(Y|A,X).\]

\paragraph{DGP2} Similarly,
since $P(A_1|L_1, A_0, X)$ and $P(A_0|X)$ do not show up in the expression of
our final plug-in estimate of the target parameters, 
we do not need to update these components. We fix the marginal distribution of baseline covariates as $P_n^0(X) = \PP_n(X)$. The two components relevant for our target parameters are $P(L_1|A_0,X)$, and $P(Y|A_1,L_1,A_0,X)$. The tangent space for DGP2 factorizes as follows:
\[L_0^2(P_n^\ell) = L_0^2(P_n^\ell(X)) \oplus L_0^2(P_n^\ell(A_0|X)) \oplus L_0^2(P_n^\ell(L_1|A_0,X)) \oplus L_0^2(P_n^\ell(A_1|L_1,A_0,X)) \oplus L_0^2(P_n^\ell(Y|A_1,L_1,A_0,X)), \]
where $L_0^2(P_n^\ell(L_1|A_0,X))$ and $L_0^2(P_n^\ell(Y|A_1,L_1,A_0,X))$ are the corresponding tangent spaces for $P_n^\ell(L_1|A_0,X)$ and $P_n^\ell(Y|A_1,L_1,A_0,X)$ respectively. At the end of each iteration, we project the chosen score $h$ into the tangent space of these relevant components, and update accordingly:
\begin{align*}
    h' = P_n^\ell(h|L_1,A_0,X) - P_n^\ell(h|A_0,X), \quad & \quad p_n^{\ell+1}(L_1|A_0,X) = (1+h')p_n^\ell(L_1|A_0,X),\\
    h'' = h - P_n^\ell(h|A_1, L_1,A_0,X), \quad & \quad p_n^{\ell+1}(Y|A_1,L_1,A_0,X) = (1+h'')p_n^\ell(Y|A_1,L_1,A_0,X).
\end{align*}

The steps outlined above project the loss-minimizing / likelihood-maximizing updates from scores in $L_0^2(P_n^\ell)$ to scores in the correct tangent spaces implied by the component in our density estimate we wish to update (i.e. $L_0^2(P(Y|A,X))$ in DGP1, $L_0^2(P(L_1|A_0, X)), L_0^2(P(Y|A_1, L_1, A_0, X))$ in DGP2).

\paragraph{Other Semi-Parametric Extensions.} Beyond projecting the updates (i.e. loss-minimizing perturbations) directly in our algorithm, analogous results (i.e. efficiency, asymptotic normality) for semi-parametric models can be directly achieved by our KDPE framework with the appropriate choice of kernel function. For closed, linear tangent spaces $T_P$, let $H_P \equiv H \cap T_P$ be the intersection of the tangent space $T_P$ and a universal RKHS $H$. The intersected space $H_P$ itself is an RKHS, with corresponding kernel $K_P$. By using this kernel for KDPE, one obtains the same results for the semiparametric models. For nonparametric models, $K_P$ was simply the kernel corresponding to the projection of $H$ onto mean-zero functions, given in Proposition 1. In semi-parametric settings, this requires a projection of $K$ onto $H \cap T_P$ in the norm of the RKHS $H$. For general semi-parametric settings, this projection for the tangent space may have to be derived analytically. While it might seem that we have traded analytic derivation of the EIF for an analytic derivation for the appropriate RKHS kernel function, the plug-in approach of KDPE works for \emph{all target parameters} that satisfy our stated assumptions 1-5. In other words, KDPE replaces the analytic derivation of the IF for each parameter of interest defined on the same semi-parametric model (requires derivation for each parameter of interest, even within the same set of assumptions) with a single analytic characterization of the RKHS projection onto the tangent space of the semi-parametric model (only requires derivation for new modeling assumptions). We chose the non-parametric setting for our paper as the simplest setting to explain KDPE, and plan to explore these extensions in future work.

\subsection{Hyperparameter Testing for KDPE}\label{app:hyperparamtable}
\paragraph{Density Bound and Convergence Criteria} In our experiments, we fixed the sample size at  $n=300$. We conducted tests with different values for the density bound $c$, with $c \in \{0.05, 0.01, 0.002\}$,
and for the stopping criteria, we tested $\gamma \in \{0.025, 0.005,0.001\}$ for DGP1, with $\lambda$ fixed at $0$.
For each combination of settings, we performed 100 simulations to obtain the results presented in Table~\ref{table:hyperparamtuning} and Figure~\ref{fig:hyperparam_dist}.

\begin{table}[H]
\centering
\begin{tabular}{rrrrrrr}
  \toprule
  $c$ & $\gamma$ & ATE RMSE & RR RMSE & OR RMSE & Avg. Iterations\\ 
  \midrule
  0.050 & 0.025  & 0.019 & 0.077 & 0.205 & 1.062 \\ 
    0.010 & 0.025  & 0.016 & 0.063 & 0.180 & 1.135 \\ 
   0.002 & 0.025  & 0.026 & 0.101 & 0.260 & 1.120 \\ 
    0.050 & 0.005 & 0.015 & 0.057 & 0.150 & 2.293 \\ 
    0.010 & 0.005  & 0.004 & 0.026 & 0.091 & 2.445 \\ 
    0.002 & 0.005  & 0.006 & 0.033 & 0.090 & 2.505 \\ 
    0.050 & 0.001  & 0.001 & 0.023 & 0.068 & 4.756 \\ 
    \bf{0.010} & \bf{0.001}  & \bf{0.009} & \bf{0.040} & \bf{0.128} & \bf{5.329} \\ 
   \bf{0.002} & \bf{0.001} &  \bf{0.009} & \bf{0.033} & \bf{0.105} & \bf{5.802} \\ 
   \bottomrule
\end{tabular}
\vspace{0.3 cm}
\caption{Average RMSE and iterations until convergence for DGP 1. No regularization, i.e., $\lambda = 0$.}
\label{table:hyperparamtuning}
\end{table}

 \begin{figure}[h]
     \centering
     \includegraphics[width=0.95\linewidth]{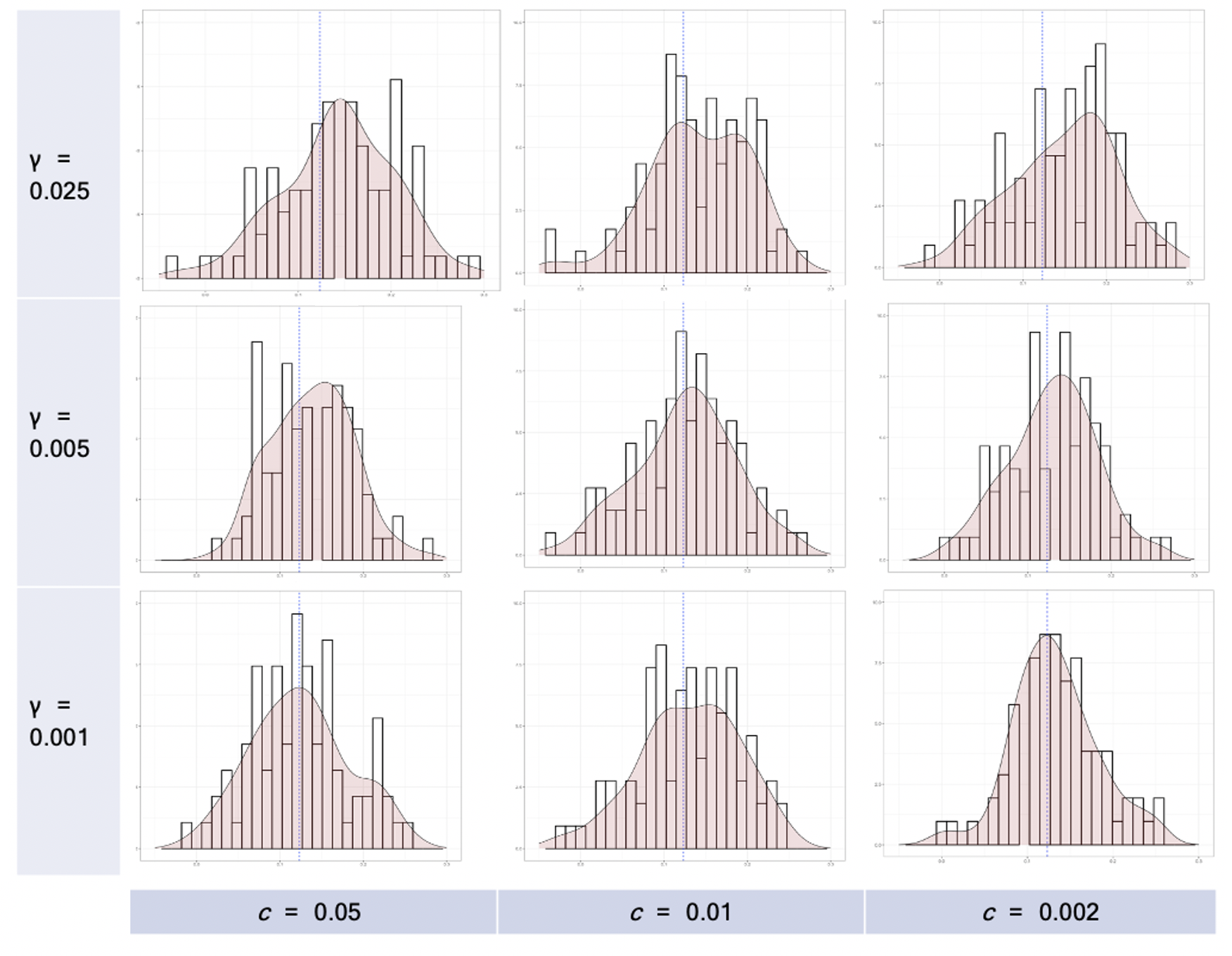}
     \caption{Density histogram of $\psi_{\ATE}(\h{P}_{\text{KDPE}})$ 
      simulations 
     as a function of $c, \gamma$ for DGP1. 
     $X$-axis: ATE values. $Y$-axis: density.
     Blue dotted line: the true ATE under $P^*$.
      , and the $y$-axis indicates the count.
     }
     \label{fig:hyperparam_dist}
 \end{figure}

The empirical results shown in Figure \ref{fig:hyperparam_dist} coincide with the theoretical results shown in Appendix~\ref{app:theorem_main}, which are derived under the setting where we have no explicit density bound $(c= 0)$ and we terminate at a fixed point $(\gamma = 0)$. As we see in Figure \ref{fig:hyperparam_dist}, the simulated distribution of $\psi_{\ATE}(\h{P}_{\text{KDPE}})$ with the smallest tested values of $c, \gamma$ obtains a distribution most similar to limiting normal distribution.
When the density bound is set to large values
(e.g., $c= 0.05$),
the resulting distributions become skewed and deviate significantly from the normal distribution.
This occurs because termination at the boundary (i.e., is tight with respect to the density bound)  prevents the distributions from satisfying 
the gradient conditions shown in Appendix~\ref{app:theorem_main}. 
When using large values for the convergence tolerance
(e.g., $\gamma=0.025$),
the algorithm terminates prematurely at a distribution that deviates significantly
from the fixed point $P_n^\infty$. Thus, to maintain the guarantees shown in Appendix~\ref{app:theorem_main} for Algorithm \ref{alg:kdpe_main}, we want the values of $(\gamma, c)$ to be as small as possible. For DGP1 in  Sections \ref{subsec:emp_dist}, we opt for $c = 0.001, \ \gamma = 0.001$, values at least as small as the tested hyperparameter values. The results of this simulation naturally pose questions of how to set $c, \gamma$ as functions of $n$, such that the residual bias caused by these hyperparameters is $o_{P^*}(1/\sqrt{n})$. This remains an open question for future work.

\paragraph{Regularization Parameter $\lambda$} While the theoretical guarantees of KDPE hold for \emph{any} $\lambda \geq 0$, the choice of $\lambda$ has practical significance, especially in relation to the convergence tolerance parameter $\gamma$. Large values of $\gamma$ (i.e., looser convergence tolerances) with heavy regularization (i.e., large $\lambda$) results in premature termination. We demonstrate these results with DGP1, testing $\lambda \in \{0, 0.001,0.01,0.1,1,10,100\}$\footnote{Note that for $\lambda=0$, which we use throughout our paper, the solver \texttt{CVXR} fails to obtain a solution roughly 15\% of the time. For those interested in testing the code, we recommend $\lambda=0.001$, which results in similar results, but avoids numerical instability issues.}. For each $\lambda$, we ran 100 simulations with $n=300$ to determine the best choice of $\lambda$. As shown in Figure \ref{fig:lambda_testing}, for a fixed value of convergence tolerance $\gamma = 0.001$, we see that larger values of $\lambda$ tend to result in skewed distributions and/or heavy tails. The average iterations for convergence corroborate that this is due to premature termination; the resulting step with heavy regularization results in a relatively small change in distribution and trivially satisfies our convergence criteria. For this reason, for DGP2 (which uses $\lambda=20$), we set the convergence tolerance to $\gamma = 0.0001$, as compared to $\gamma = 0.001$ in the non-regularized case in DGP1. We keep the density bound for DGP2 the same ($c = 0.001$). Future extensions of this work include how to set regularization parameter $\lambda$ as a function of convergence tolerance $\gamma$. 

\begin{figure}
    \centering
    \includegraphics[scale=0.4]{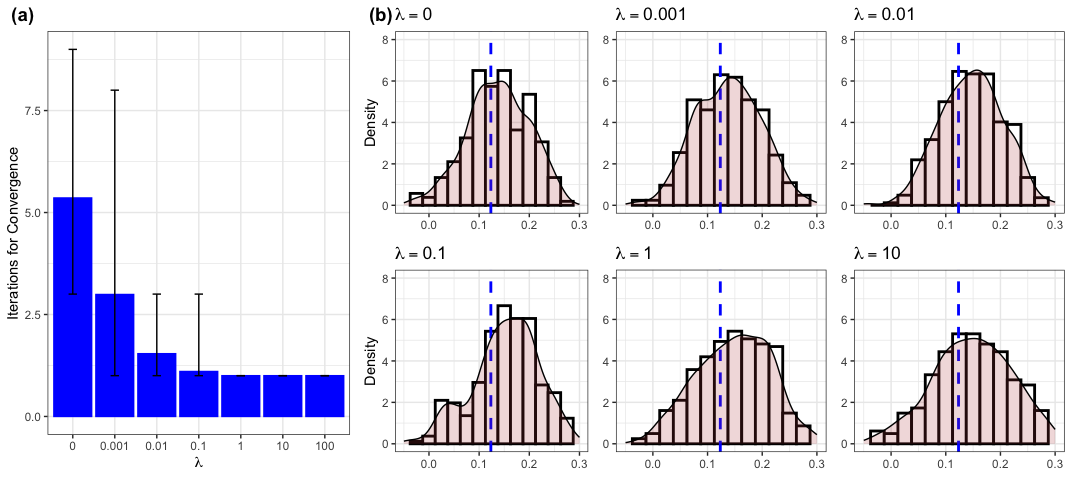}
    \caption{(a): Average number of iterations, with maximum and minimum number of iterations for each $\lambda$ tested ;(b): Distribution of $\h\psi^{\ATE}_{\KDPE}$ in DGP1 for different $\lambda$, 
$c = 0.002, \gamma=0.0001$. Blue line denotes true value. 100 simulations for each setting of $\lambda$. }
    \label{fig:lambda_testing}
\end{figure}

\subsection{Additional (L)TMLE Results for DGP2}\label{append:DGP2}
We provide additional results to show that the poor performance of TMLE in DGP2 is due to the relatively small sample size $n=300$. For $n=800$, Figure \ref{fig:enter-label} see that the distribution of $\h{\psi}^{\ATE}_{\TMLE}$ is far closer to the limiting distribution (shown in purple) than when $n=300$, as shown in the main body of our paper. This indicates that the performance of TMLE for DGP2 at $n=300$ is indeed due to poor finite sample performance.

\begin{figure}[htbp!]
    \centering
    \includegraphics[scale=0.5]{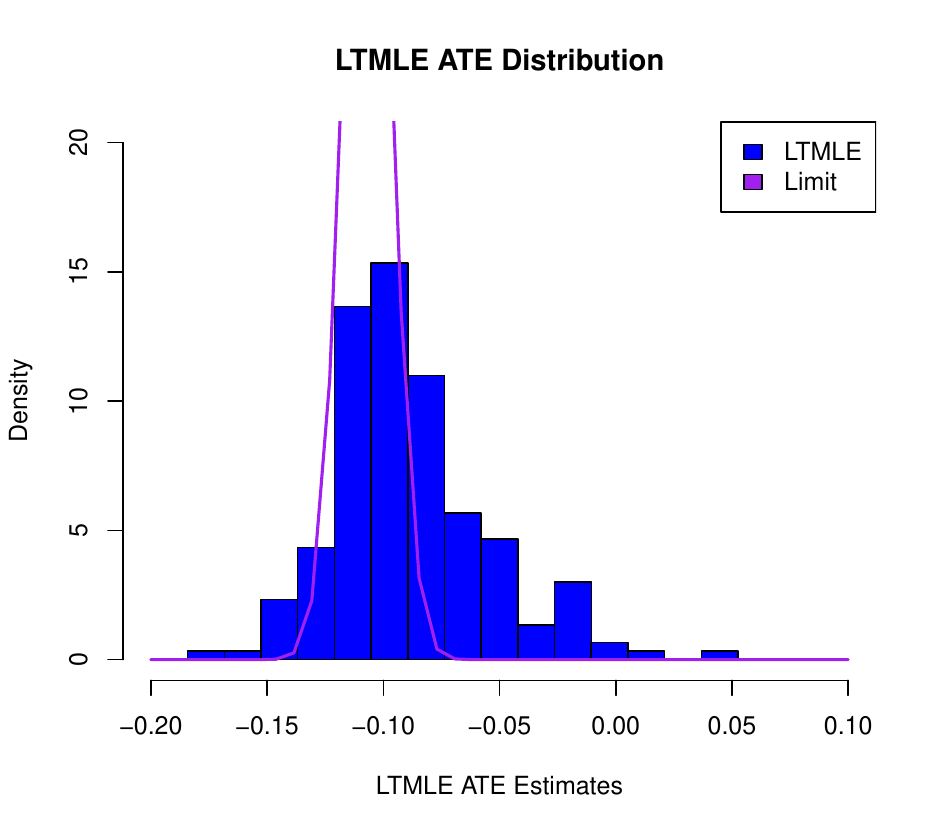}
    \caption{Distribution of $\h\psi_{\TMLE}^\ATE$ for $n=800$, using LTMLE Package for DGP2. Purple shows correct limiting distribution.}
    \label{fig:enter-label}
\end{figure}

\subsection{Additional Empirical Results for Relative Risk and Odds Ratio}

To show that our asymptotic normality holds for multiple target parameters $\psi$, we include the simulated distributions of the KDPE estimate against TMLE and SL in Figure \ref{figure: asymp_normal}, using the same hyperparameter settings described above. 

\begin{figure}[t]
    \centering
    \includegraphics[scale = 0.5]{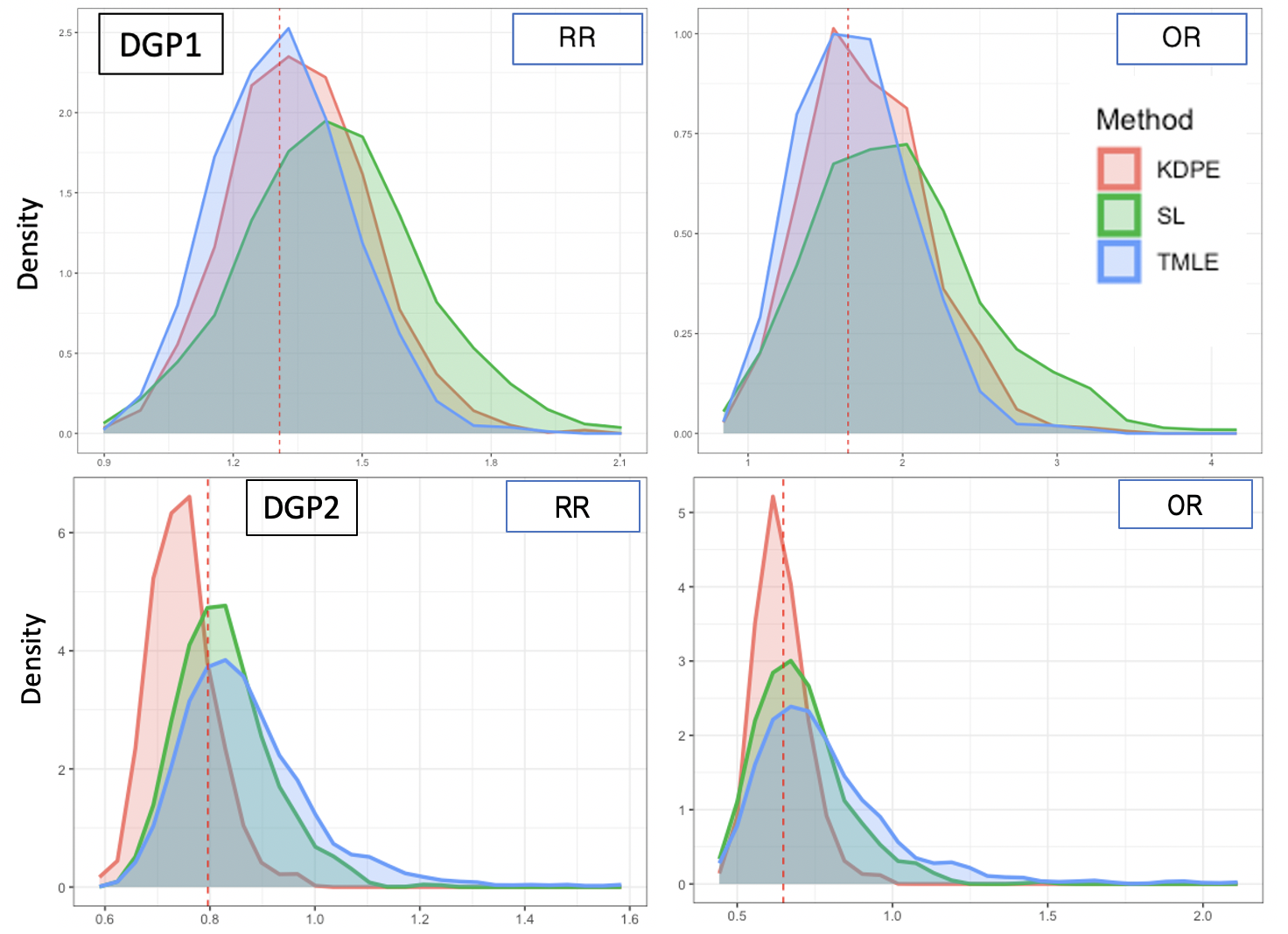}
    \caption{Simulated distributions of $\h \psi_{\mathrm{RR}}, \h \psi_{\mathrm{OR}}$. First row corresponds to DGP1, and second row to DGP2. Red line denotes true value of the target parameter.}
    \label{figure: asymp_normal}
\end{figure}

\subsection{Runtime Complexity for KDPE}
Our algorithm is a meta-algorithm, and thus the computational complexity is dependent on different factors (integration method, optimization software, number of iterations, etc.). We provide computational complexity results per iteration in the case of Algorithm \ref{alg:kdpe_main}. The set-up necessary to formulate the optimization problem requires us to obtain $K_P(O_i, O_j)$ (Eq. 20) for all $i,j \in [n]$, and occurs the cost $\mathcal{O}(n^2)$:
\begin{itemize}
    \item Computing $\exp(-\|O_i-O_j\|^2)$ for all $i,j \in [n]$: $\mathcal{O}(n^2)$
    \item Computing $f_P(O_i)$ for all $i \in [n]$: $\mathcal{O}(4n^2)$
    \item Computing $c_P$: $\mathcal{O}(4n)$
\end{itemize}
The optimization problem for KDPE involves $n$ regressors, and therefore the computational complexity of this method is $\mathcal{O}(n^2 + k(n))$ per iteration, where $k(n)$ is the complexity for the solver.

In contrast, TMLE (as introduced in the paper) occurs an $O(n)$ cost when calculating the influence function in each iteration, and involves only 1 regressor, resulting in a complexity of $\mathcal{O}(n + k(1))$ in each iteration. 

The difference in cost per iteration (quadratic, as opposed to linear in $n$), can be seen as the price to pay for a plug-in distribution that works with all parameters satisfying our regularity assumptions.  Our empirical results corroborate these scaling results. The median time-per-iteration in DGP1 for 150 samples and 300 samples is 0.22 and 0.91 respectively, indicating a roughly $\mathcal{O}(n^2)$ scaling for our method. For DGP1, we average between 5-6 iterations before termination, using the hyperparameters in Section \ref{sec:sims}. For DGP2, we average 2 iterations before termination. Across our simulations for DGP1, the average runtime for our method is 5.8 seconds, with 5.8 iterations on average. TMLE, as implemented in \cite{one_step_tmle}, uses the one-step method,  which performs the MLE optimization in one iteration using logistic regression. Because this version of TMLE involves no iterations (same time complexity as logistic regression), it is a poor comparison with KDPE in terms of time complexity.  The results of TMLE are included as a standard of comparison for the distribution of our estimator, rather than its time complexity.

\subsection{Bootstrap Algorithm}\label{app:bootstrap}
We use the classical bootstrap procedure \citep{boostrap_ref}, i.e.,  sampling $n$ observations from our observed data $\{o_i\}_{i=1}^n$ with replacement,
to estimate
the variance of the estimator. 
We formally state this procedure in
Algorithm \ref{alg:bootstrap}.

\begin{algorithm}
    \caption{Bootstrap Confidence Interval for KDPE}
    \label{alg:bootstrap}
    \begin{algorithmic}[1]
    \STATE \textbf{Input}: $\{O_i\}_{i=1}^n$, convergence tolerance function $l(\cdot)$, convergence tolerance $\gamma$, density bound $c$, and bootstrap iterations $m$. 
    \STATE Obtain distribution $\h{P}_{\KDPE} = \KDPE(\{O_i\}_{i=1}^n,l(\cdot), \gamma, c)$. 
    \STATE Set the estimate $\h{\psi}=\psi(\h{P}_{\KDPE})$.
    \FOR{$j \in \{1,...,m\}$}
        \STATE Sample with replacement $\{O_i^*\}_{i=1}^n$ from sample $\{O_i\}_{i=1}^n$, and remove duplicates.
        \STATE Set the estimate $\h{\psi}_j=\psi(\text{KDPE}(\{O_i^*\}_{i=1}^n,l(\cdot), \gamma, c))$
    \ENDFOR
    \STATE Calculate the mean estimate $\bar{\psi} = \sum_{j=1}^n \h\psi_j/n.$
    \STATE Estimate the variance $\h{S} = \sum_{j=1}^m(\h{\psi}_j-\bar{\psi})^2/(n-1)$. 
    \STATE Return $\left[\h\psi \pm \Phi^{-1}(1-\alpha)\sqrt{\h{S}}\;\right]$.
    \end{algorithmic}    
\end{algorithm}

\subsection{Bootstrap Intervals for Inference}\label{subsec:bootstrap} 
 %
 %
\begin{table}[h]
\centering
\begin{tabular}{rlrrrr}
\toprule
    & Parameter & $\h{\textsf{Var}}(\h\psi_{\text{target}})$ & Avg. Length (95\% CI) & Coverage (95\% CI) \\ 
  \midrule
   KDPE & $\psi_\ATE$  & 0.00427 (0.00287) & 0.252 & 0.945\\ 
        & $\psi_\RRm$ & 0.04562 (0.02555)& 0.826 &   0.975\\ 
        & $\psi_\ORr$ & 0.30303 (0.15322) & 2.107&  0.970\\
 \midrule
 TMLE & $\psi_{\ATE}$ & 0.00192 (0.00296) &0.172  &  0.903\\ 
      & $\psi_\RRm$  & 0.01483 (0.02837)  & 0.510  & 0.903\\
      & $\psi_\ORr$  & 0.07016 (0.14981)  &1.230  & 0.907\\
 \bottomrule
\end{tabular}
\caption{Estimated Variance and Confidence Intervals of $\h\psi_{\text{ATE}}$, $\h\psi_{\text{RR}}$, $\h\psi_{\text{OR}}$ of KDPE and TMLE for DGP1 across 237 simulations, with 100 bootstrap iterations each for KDPE.}
\label{table:bootstrap}
\end{table}

Table \ref{table:bootstrap} reports average estimated variance, average length of 95\% intervals, and the coverage across 237 simulations of $n=300$ for DGP1. For comparison, we include the results of inference for TMLE, which estimates the variance $\h{\textsf{Var}}(\h\psi) = \PP_n[(\phi_{\h P_{\TMLE}}^{\text{target}})^2]$ using knowledge of the (efficient) IF. The baselines for both TMLE and KDPE (in parentheses) are the sample variances of the simulated distributions in Section \ref{subsec:emp_dist}, and provided next to average estimated variance. The KDPE-bootstrap variance estimates overestimate the variance for all target parameters relative to the KDPE baseline; the inflated variance estimates lead to conservative confidence intervals for the bootstrapped KDPE estimates with over-coverage, as shown in the coverage results of Table \ref{table:bootstrap}. By contrast,  TMLE's IF-based variance estimate underestimates the variance for all target parameters relative to the TMLE baseline, leading to undercoverage. These preliminary results suggest that bootstrap interval calibration is future direction for improvement, and are intended as a starting point for future investigation. While these results imply that our bootstrap-estimated confidence intervals are conservative, more extensive testing is required to empirically validate these results. A large theoretical gap also remains for confirming that our bootstrap-estimated variances are consistent asymptotically, and providing an analysis for Algorithm \ref{alg:bootstrap} under KDPE is another natural step towards our goal of fully computerized inference. Lastly, as shown in Table \ref{table:bootstrap}
and Algorithm \ref{alg:bootstrap}, we only use the bootstrap samples to estimate the variance of our estimator and use a normal approximation to construct confidence intervals. Alternative methods for constructing confidence intervals include directly taking the quantiles of the bootstrap estimates and/or subsampling while bootstrapping, which we do not explore in this project.